\documentclass[pdflatex,sn-mathphys-num]{sn-jnl}

\usepackage{graphicx}%
\usepackage{multirow}%
\usepackage{amsmath,amssymb,amsfonts}%
\usepackage{amsthm}%
\usepackage{mathrsfs}%
\usepackage[title]{appendix}%
\usepackage{xcolor}%
\usepackage{textcomp}%
\usepackage{manyfoot}%
\usepackage{booktabs}%
\usepackage{algorithm}%
\usepackage{algorithmicx}%
\usepackage{algpseudocode}%
\usepackage{listings}%
\usepackage{stmaryrd}
\usepackage{mathtools}
\usepackage{bm}
\usepackage{enumitem}
\usepackage{geometry}
\usepackage{braket}
\usepackage{array}
\usepackage{longtable}
\usepackage{float}

\AtBeginDocument{\fontsize{9.5}{11}\selectfont}

\theoremstyle{thmstyleone}
\newtheorem{theorem}{Theorem}[section]
\newtheorem{proposition}[theorem]{Proposition}

\theoremstyle{thmstyletwo}
\newtheorem{lemma}[theorem]{Lemma}
\newtheorem{corollary}[theorem]{Corollary}

\theoremstyle{thmstylethree}
\newtheorem{definition}[theorem]{Definition}

\DeclareMathOperator{\Sym}{Sym}
\DeclareMathOperator{\Aut}{Aut}

\DeclareMathOperator{\diag}{diag}
\DeclareMathOperator{\Tr}{Tr}

\DeclareMathOperator{\End}{End}

\DeclareMathOperator{\rank}{rank}

\begin{document}

\title{Finite Standard-Model Carrier from Timelike Codazzi Defects}

\author*[1]{\fnm{Piotr} \sur{Ogonowski}}\email{piotrogonowski@kozminski.edu.pl}
\affil*[1]{\orgname{Kozminski University}, \orgaddress{\street{Jagiellonska 57/59}, \city{Warsaw}, \postcode{03-301}, \country{Poland}}}

\abstract{
The inverse problem of whether protected data carried by a local Lorentzian defect can determine a finite internal gauge-and-matter carrier before microscopic dynamics are chosen is considered. For a codimension-three timelike Codazzi defect, a marking-preserving reconstruction fixes a five-dimensional complex carrier $V=C\oplus W$ with $\dim_{\mathbb C}(C,W)=(3,2)$, the global gauge form $S(U(3)\times U(2))$, a chiral $16$-dimensional exterior packet with $k_Y=5/3$, and a norm-seven family structure. The chiral packet has the standard algebraic $SU(5)$-Spin$(10)$ envelope without assuming grand unification. Within a declared regularized-domain class, one constrained microscopic realization is constructed. Its distinguished branch ends at a single intrinsic terminal where the strict charge minimum and the first simple invariant-sector fixed-charge degeneracy coincide; soft-coupled readouts acquire a rank-one critical-response pole. On the rank-two exact-five moment stratum, four native matrix moments reconstruct a five-dimensional cyclic operator up to hidden unitary. The resulting scale-free observables lie on a one-dimensional solution manifold with a discrete conjugate orientation. With $W/Z$ used to set its longitudinal coordinate, direct CKM and reactor/solar data are consistent with the same manifold, while first-row CKM normalization remains a separate transverse structural test. Existence within the stated class is established, not uniqueness. A blind determination of the longitudinal coordinate still requires the unresolved centered-response incidence; physical transport remains an independent microscopic gate.
}

\keywords{Lorentzian geometry, Codazzi tensors, gauge reconstruction, spectral geometry, non-Abelian Hodge theory, constrained dynamics, fixed-charge stability, Standard Model, flavor structure, Alena Tensor}

\maketitle
\clearpage

\section{Introduction}
\label{sec:introduction}

Gauge and matter structure is usually introduced as part of the microscopic model. Geometric encodings include Kaluza-Klein mechanisms \cite{Helein2022KaluzaKleinMechanisms}, Eisenhart-Duval lifts \cite{Cariglia2015EisenhartDidactical}, Randers-Finsler descriptions \cite{Silva2021RandersFinslerFieldTheory}, and recent geometric treatments of gauge internal spaces \cite{Gomes2024InternalSpaces}. Internal symmetries and chiral couplings in Kaluza-Klein models are represented by \cite{Baptista2024InternalSymmetries} and \cite{Baptista2025ChiralKaluzaKlein}; finite-resolution curvature defects provide a local comparison class \cite{Czuchry2026}. Defect-operator frameworks are represented by \cite{BilloGoncalvesLauriaMeineri2016DefectsCFT}, and gauge-theory defects as probes by \cite{DrukkerGaiottoGomis2011VirtueDefects}. The inverse problem considered here is whether a resolved local defect can force a finite internal carrier before microscopic dynamics is specified.

A local defect supplies protected boundary data rather than a microscopic gauge-and-matter model. Three logically distinct levels must therefore be separated: structural reconstruction from the boundary source, realization by one constrained Parent dynamics, and low-energy completion after spectral attachment. For a codimension-three timelike defect the oriented normal bundle has link $S^2\simeq\mathbb{CP}^1_\Gamma$, and removal of the scalar trace leaves the order-one current and order-two trace-free Codazzi grades. Standard $SU(2)$ representation theory is represented by \cite{FultonHarris1991RepresentationTheory}. The reconstruction criterion used below retains only data that admit marking-preserving recovery on the protected quotient; its formal definition is given after the geometric defect data have been fixed.

The theorem-level result has two stages. First, the protected defect data determine a least-sufficient finite carrier before an action is chosen: a primitive $3+2$ complex, the global gauge form, a chiral packet, and the norm-seven family structure are fixed on the protected quotient. The exterior packet has the standard derived $SU(5)$-Spin$(10)$ representation envelope without a grand-unified gauge assumption. Pure-spinor/twistor \cite{Woit2021} and finite permutation/Clifford \cite{GresnigtGourlayVarma2023} constructions remain comparisons; recent algebraic three-generation realizations are represented by \cite{GourlayGresnigt2024}, with a broader survey in \cite{Furey2025}. The marked spectral trinion uses the correspondence of \cite{BeauvilleNarasimhanRamanan1989Spectral} and the real-component classification of \cite{HellerOuyangPedit2024SYZ}. Second, one same-domain Alena-Codazzi Parent realization is constructed without finite-jet retuning. Its selector branch has a unique first intrinsic critical terminal on the constructed interval. On the rank-two exact-five moment stratum, four native matrix moments reconstruct the minimal five-dimensional cyclic attachment up to hidden unitary.

The remaining microscopic statements are formulated as explicit gates rather than assumed identifications. Actual root and norm-seven Wilson transport, the $q=5$ and essential gaps, the normalized centered response, absolute normalization, and continuum/outgoing completion remain to be exported. After the structural packet, Parent realization and microscopic attachment have been fixed, the scale-free observables form a one-dimensional solution manifold. One predeclared chart datum selects the longitudinal member identified with the Standard-Model realization of the same general action; independently sourced transverse relations are no-retune tests, and an independently certified Parent response would replace this calibration by an additional longitudinal test.

The dependency graph is summarized in Figure~\ref{fig:logical-structure}. Section~\ref{sec:filtered-self-reconstruction} states the defect data and reconstruction principle. Section~\ref{sec:primitive-gauss-local-self-reconstruction} derives the finite internal carrier and its gauge, chiral and family descendants. Section~\ref{sec:alena-codazzi-realization} constructs one constrained microscopic realization and its attachment gates. Section~\ref{sec:quantitative-branch-tests} evaluates the resulting one-dimensional solution manifold and its holdouts. Section~\ref{sec:conclusions} summarizes the result and the remaining microscopic completion.

\begin{figure}[H]
\centering
\fbox{\begin{minipage}{0.92\textwidth}
\centering
protected Codazzi defect data
$\longrightarrow$ finite $3+2$ carrier and structural descendants\\
$\Downarrow$\\
one constrained Parent realization
$\longrightarrow$ intrinsic terminal\\
$\Downarrow$\\
moment attachment and transport/response gates
$\longrightarrow$ conditional scale-free readout\\
$\Downarrow$\\
one longitudinal calibration \quad$|$\quad transverse holdouts \quad$|$\quad downstream completion
\end{minipage}}
\caption{Logical structure and status of the construction. Reconstruction of the finite carrier and existence of one Parent realization are theorem-level results. Transport, response and continuum incidences are explicit microscopic gates; phenomenological comparisons are evaluated only after the reduced manifold has been fixed.}
\label{fig:logical-structure}
\end{figure}

\section{Timelike Codazzi Defect and the Reconstruction Problem}
\label{sec:filtered-self-reconstruction}

\subsection{Geometric defect data and source grading}

Let $\Gamma$ be a smooth embedded future-directed timelike worldline in an oriented time-oriented Lorentzian four-manifold. Its oriented Euclidean normal bundle has rank three and the real blow-up has boundary fiber $S^2_\Gamma\simeq\mathbb{CP}^1_\Gamma$. Codazzi and Lorentzian optical conventions are represented by \cite{DerdzinskiShen1983Codazzi}, \cite{ManticaMolinari2023CodazziSpaceTimes}, and \cite{FerrandoSaez2007ShearFreeGeodesic}; recent Codazzi geometry is represented by \cite{ReberTerek2024}. Spinor/twistor conventions follow \cite{PenroseRindler1984v1}, \cite{PenroseRindler1986v2}, and \cite{WardWells1990}. Rainich reconstruction is represented by \cite{Rainich1925Electrodynamics}; first-order and twistor comparisons by \cite{Plebanski1977} and \cite{MasonWoodhouse1996}. Einstein-Weyl/minitwistor comparisons use \cite{JonesTod1985Minitwistor} and \cite{CalderbankPedersen2000EinsteinWeyl}. The atomic positive linking class is
\begin{equation}
c_1(L_\Gamma)=1,\qquad L_\Gamma\simeq\mathcal O(1).
\label{eq:pgl-primitive-line}
\end{equation}
Higher classes are composite filtered data. Regularity conventions are represented by \cite{Uhlenbeck1982RemovableSingularities} and \cite{GilbargTrudinger2001Elliptic}; spherical-tensor conventions by \cite{Varshalovich1988QuantumAngularMomentum} and \cite{Thorne1980MultipoleExpansions}. After scalar traces are removed, the retained transverse source through second normal order is
\begin{equation}
\operatorname{gr}_{>0}\mathcal J_{\perp,{\rm ns}}^{\leq2}\simeq V_1[1]\oplus V_2[2].
\label{eq:pgl-two-jet-type}
\end{equation}

The organizing reconstruction loop is
\begin{equation}
\Phi\longmapsto\left(\mathcal A_\Phi,\mathfrak S(\mathcal A_\Phi),\delta_\Phi,\operatorname{Sec}_{\rm iso}(\mathcal A_\Phi),\nabla_\Phi\right)\longmapsto\widehat\Phi.
\label{eq:intro-self-reconstruction-loop}
\end{equation}
Equality is taken on the protected physical quotient after gauge, diffeomorphism and marking-preserving unitary equivalence.

\subsection{Protected least-sufficient reconstruction}

The loop \eqref{eq:intro-self-reconstruction-loop} is used with structural closure preceding microscopic attachment and Schur transfer. Operationally, PLSR retains the minimal marked data required to recover the protected quotient and removes multiplicities or refinements that do not change that recovery. At depth $i$ retain $X_{\leq i}(\Phi)=(\mathcal A_\Phi;\mathfrak A_{\leq i}^{\rm rec},\pi_i,\mathcal H_i,\mathcal M_{\leq i})$, with $\mathscr R\mathscr O=\operatorname{id}$ on the protected quotient and Gram Hessian $H_i^{\rm Gram}=DR_i^\dagger G_iDR_i$. Spectral/Fredholm conventions follow \cite{Kato1995}, \cite{Callias1978}, and \cite{Anghel1993Callias}; curvature and superselection follow \cite{Besse2007} and \cite{DoplicherHaagRoberts1971LocalObservablesI}; determinant/index conventions follow \cite{BismutFreed1986EllipticFamiliesII}, \cite{FreedLott2010DifferentialKIndex}, and \cite{Quillen1985Superconnections}.\\
~\\
For active principal types $\Lambda_{\rm pr}$ the marked central sector and its fidelity defect are
\begin{equation}
\zeta_F:\bigoplus_{a\in\Lambda_{\rm pr}}\mathbb Ce_a\longrightarrow Z(\mathfrak A_F),\qquad \zeta_F(e_a)=p_a,
\label{eq:fsr-charge-reconstruction-map}
\end{equation}
\begin{equation}
\delta_{\rm fid}(F)=\dim\ker\zeta_F+\sum_{a\in\Lambda_{\rm pr}}\left(\dim\ker\iota_{F,a}+\|\rho_{F,a}\iota_{F,a}-\operatorname{id}_{V_a}\|_{\rm HS}^2\right)\geq0.
\label{eq:fsr-charge-fidelity-defect}
\end{equation}
Reconstructive completeness is
\begin{equation}
\pi_i(\mathfrak A_{\leq i}^{\rm rec})'=\pi_i\!\left(Z(\mathfrak A_{\leq i}^{\rm rec})\right).
\label{eq:fsr-commutant-closure}
\end{equation}
At a protected interface $X\preceq_{\rm pp}Y$ means marking-preserving factorization of the observation of $X$ through that of $Y$, commuting with all retained operations; mutual factorization through retained isomorphisms defines protected equivalence.

\begin{definition}[Protected least-sufficient reconstruction]
\label{def:fsr-plsr}
A realization is sufficient when marking-preserving observation and recovery satisfy $\mathscr R\mathscr O=\operatorname{id}$ on the protected physical quotient. It is least-sufficient when every sufficient realization descends to it in $\preceq_{\rm pp}$ while preserving the declared grading, $\ast$/Real or integral structure, generated products, exact kernels/cokernels, retained dynamics, and gapped functional calculus. Canonically generated source-target morphisms are retained; unprotected multiplicities, spectral refinements, finite-index enlargements, and flat twists are removed whenever sufficiency is preserved. The finite reconstruction core is the faithful marked representation of the resulting least protected extension satisfying \eqref{eq:fsr-commutant-closure}, modulo marking-preserving unitary equivalence.
\end{definition}

\begin{proposition}[Protected reduction and finite-core uniqueness]
\label{prop:fsr-karoubi-localization}
\label{prop:fsr-exact-contractive-recovery}
\label{prop:fsr-approximate-contractive-recovery}
\label{prop:fsr-finite-core-uniqueness}
\label{prop:fsr-no-unsourced-multiplicity}
\label{prop:fsr-protected-output-descent}
For a split recovery pair $(\mathscr O_X,\mathscr R_X)$, $e_X:=\mathscr O_X\mathscr R_X$ represents the retained Karoubi image. Exact non-expansive recovery makes $\mathscr O$ isometric and $\mathscr O\mathscr R$ the orthogonal projector onto $\operatorname{ran}\mathscr O$; if $\|\mathscr R\mathscr O-\operatorname{id}\|_{\rm op}\leq\varepsilon<1$, standard polar-decomposition estimates give stable metric recovery with $O(\varepsilon)$ distortion. For $\mathfrak A\simeq\bigoplus_aM_{n_a}(\mathbb C)$ with labeled multiplicities $m_a\geq1$, the commutant excess is $\Delta_{\rm mult}(\pi)=\sum_a(m_a^2-1)$. Hence \eqref{eq:fsr-commutant-closure} is equivalent to $m_a=1$ for every labeled simple block, and the faithful finite core is unique up to marking-preserving unitary equivalence. Gapped finite words inherit the corresponding perturbative control through Proposition~\ref{prop:bl-gapped-schur-stability}.
\end{proposition}
A canonical integral or cyclotomic form may be imposed before complexification; a commensurate downstream corner retains only the least admissible finite-index enlargement. For a graded source $Q=\bigoplus_{\ell\in\Lambda}Q_\ell$ and reconstructed level operator $N_F$, degree fidelity is $\zeta(e_\ell)={\bf1}_{\{\ell\}}(N_F)$.

For a finite typed residual ledger $\mathscr L=\{R_\alpha\}_{\alpha\in A}$,
\begin{equation}
\mathfrak C(\mathscr L):=\bigcap_{\alpha\in A}\{R_\alpha=0\}=\left\{\sum_{\alpha\in A}w_\alpha\|R_\alpha\|^2=0\right\},\qquad w_\alpha>0.
\label{eq:fsr-strict-saturation-stratum}
\end{equation}
On a smooth one-dimensional Parent branch, mandatory residuals are reduced only by proven dependencies. Derived readings follow from already passed objects; same-object incidences join independently typed realizations; holdouts are excluded from construction of the tested object; completion data are downstream of protected reduction. Evidential codimension is assigned only to independent provenance. A residual consuming the physical branch direction is selector-active, while positive gap or transversality margins are monitored after the relevant closed incidences have passed. This evidence graph is used below.\\
~\\
The comparison class includes the Robertson-Schr\"odinger bound \cite{Robertson1929Uncertainty} and \cite{Schrodinger1930Uncertainty}, conditional expectations \cite{Takesaki1972ConditionalExpectations}, data processing \cite{Lindblad1976QuantumDynamicalSemigroups} and \cite{Lindblad1975EntropyInequalities}, and quantum sufficiency \cite{Petz1988SufficiencyChannels}, \cite{JencovaPetz2006Sufficiency}, and \cite{Kuramochi2017MinimalSufficient}. Approximate/subalgebra recovery is represented by \cite{GaoJungeLaRacuente2020SubalgebraEntropy}, \cite{JungeRennerSutterWildeWinter2018Recovery}, \cite{GaoJungeLaRacuenteLi2025EntropyDecay}, and \cite{Wirth2025ExponentialEntropyDecay}. These frameworks constrain recovery but do not select a Parent term or normalization.

\section{Finite Internal Carrier Reconstructed from the Defect}
\label{sec:primitive-gauss-local-self-reconstruction}

The structural reconstruction is organized around the defect data of Section~\ref{sec:filtered-self-reconstruction}. PLSR and Borel-Weil minimality fix $W=E_2$, $C=E_3\simeq\Sym^2W$, and the ordered multiplication $\mu:\Sym^2W\to C$. Three descendant maps are then used: the primitive exact complex gives the port/Green geometry, the determinant/exterior functor gives the gauge and chiral packet, and the integral $A_2$ datum gives the family packet. Their microscopic realization is deferred to Section~\ref{sec:alena-codazzi-realization}; independently reconstructed readings are identified only through explicitly typed incidences.

\subsection{Primitive $A_2$ source and marked separation}
\label{subsec:primitive-gauss-local-data}

\begin{theorem}[Principal-$A_2$ source envelope and marked separation]
\label{prop:pgl-source-a2-identification}
With $V_2\simeq\operatorname{Sym}^2_0(N_\Gamma)$, the compact and split readings are
\begin{equation}
\Psi_\Gamma:V_1\oplus V_2\longrightarrow\mathfrak{su}(N_\Gamma^{\mathbb C}),\qquad
\Psi_\Gamma(u,B)=[u]_\times+iB,
\label{eq:pgl-source-a2-identification}
\end{equation}
and $\mathfrak{sl}_3(\mathbb R)=\mathfrak{so}(N_\Gamma)\oplus\operatorname{Sym}^2_0(N_\Gamma)$. Their common complexification is $\mathfrak{sl}_3(\mathbb C)$ with $[V_1,V_1]\subset V_1$, $[V_1,V_2]\subset V_2$, $[V_2,V_2]\subset V_1$. Under principal $SU(2)$ the adjoint is $V_1\oplus V_2$ \cite{Kostant1959PrincipalThreeDimensional}, selecting $A_2$ among simple rank-two systems. The trace-free second Lorentz jet has the same $V_0\oplus V_1\oplus V_2$ restriction, so the retained non-scalar grades are exactly \eqref{eq:pgl-two-jet-type}.\\
~\\
For unit $e\in N_\Gamma$ put
\begin{equation}
Q_e=e\otimes e-\frac13{\bf1},\qquad
\mathcal V(e)=\sqrt{\frac32}Q_e,\qquad
\mathcal V(-e)=\mathcal V(e),\qquad \mathcal V^*g_{S^4}=3g_{S^2}.
\label{eq:pgl-veronese-source-orbit}
\end{equation}
For $I_2=\Tr(B^2)$ and $I_3=\Tr(B^3)$,
\begin{equation}
\chi_B(\lambda)=\lambda^3-\frac{I_2}{2}\lambda-\frac{I_3}{3},\qquad
\Delta_{A_2}=\frac12I_2^3-3I_3^2,\qquad
\det\mathcal G_{\rm inv}=12\Delta_{A_2}.
\label{eq:pgl-a2-orbit-discriminant}
\end{equation}
Thus $(I_2,I_3)$ fixes the $SO(3)$ orbit and $\Delta_{A_2}=0$ is the repeated-eigenvalue/invariant-rank-loss locus. The normalized determinant cubic is the standard $g=3$ Cartan isoparametric system; \eqref{eq:pgl-veronese-source-orbit} is its extremal focal orbit, with minimality represented by \cite{Takahashi1966MinimalImmersions}. This source geometry does not fix a physical transverse stiffness.\\
~\\
For the simply connected compact form, the Coxeter-compatible maximal-torus lattice is
\begin{equation}
\Lambda_E:=Q(A_2)\simeq\mathbb Z[\omega],\qquad E_\omega:=\mathbb C/\Lambda_E.
\label{eq:pgl-a2-eisenstein-lattice}
\end{equation}
The source envelope precedes finite marked reconstruction and does not merge the two protected labels.
\end{theorem}

The invariant quotient in \eqref{eq:pgl-a2-orbit-discriminant} also carries its standard local Coxeter-Saito structure. With $u_F:=I_2/6$ and $w_F:=-I_3/9$, the flat pencil fixes, modulo quadratic terms, the polynomial Frobenius potential
\begin{equation} \mathcal F_{A_2}(u_F,w_F)=\frac32u_Fw_F^2+\frac18u_F^4,\qquad e=\partial_{w_F},\qquad E=w_F\partial_{w_F}+\frac23u_F\partial_{u_F}. \label{eq:pgl-a2-frobenius-potential} \end{equation}
For $v=\partial_{u_F}$ one has $v\circ v=u_Fe$. The corresponding Landau-Ginzburg polynomial is $p(r)=r^3-3u_Fr+3w_F$: $w_F$ is the unity deformation and $u_F$ fixes the critical-value splitting. Thus $4u_F^3-9w_F^2=0$ is the zero of one canonical coordinate while the tangent Frobenius algebra remains semisimple for $u_F\neq0$. The tangent Frobenius caustic is $\{u_F=0\}$ and meets this divisor only at its apex $(u_F,w_F)=(0,0)$. The distinction is retained below when the selected nonzero fold is interpreted. ADE Frobenius/Landau-Ginzburg duality is represented by \cite{BriniMaStrachan2025}; finite complex-reflection/Saito reductions are represented by \cite{KonishiMinabe2025}.

\begin{lemma}[Marked Gauss-local two-label separation]
\label{lem:pgl-reconstructive-charge-sector-separation}
\label{cor:pgl-one-block-alias-exclusion}
The marked Casimir projectors are
\begin{equation}
\Pi_1=\frac{6{\bf1}-\mathcal C_\partial}{4},\qquad
\Pi_2=\frac{\mathcal C_\partial-2{\bf1}}{4}.
\label{eq:pgl-principal-casimir-projectors}
\end{equation}
They are central in a source-free relative annulus. Nonzero recovered charges with protected left inverses therefore give two orthogonal central labels carrying faithful marked $V_1$ and $V_2$ copies; the source Euler operator is $\mathsf N_\partial=(\mathcal C_\partial+2{\bf1})/4$. In particular a single $M_3(\mathbb C)$ block cannot represent both nonzero labels.
\end{lemma}

The relative Gauss-local algebra is interpreted in the protected-superselection sense of \cite{DoplicherRoberts1990}. The primitive input is therefore \eqref{eq:pgl-primitive-line}, oriented $N_\Gamma$, the two marked source grades, and Gauss-local recovery of their labels. The one-channel and higher/composite branches remain rigidity comparisons and are not used to infer the primitive $3+2$ carrier.

\subsection{Exact complex, port geometry and rigidity}
\label{subsec:universal-rank-five-support}

The marked grades are realized by the canonical section ring, with $E_q\simeq\operatorname{Sym}^{q-1}\mathbb C^2$ by Borel-Weil \cite{Bott1957HomogeneousVectorBundles}; finite Toeplitz conventions follow \cite{BordemannMeinrenkenSchlichenmaier1994} and \cite{Schlichenmaier2010BerezinToeplitzReview}. Hence $W=E_2$, $C=E_3\simeq\operatorname{Sym}^2W$ and
\begin{equation}
V_\Gamma=C_\Gamma\oplus W_\Gamma=E_3\oplus E_2,
\qquad
\mathfrak A_\Gamma^{\rm src}\simeq M_3(\mathbb C)\oplus M_2(\mathbb C).
\label{eq:pgl-rank-five-carrier}
\end{equation}
\begin{equation}
\mathfrak A_\Gamma^{\rm src}:=C^*(p_C,p_W,\iota_2(V_2),\iota_1(V_1))
=\operatorname{End}(C_\Gamma)\oplus\operatorname{End}(W_\Gamma).
\label{eq:pgl-source-generated-finite-core}
\end{equation}

\begin{theorem}[Source-generated primitive packet]
\label{thm:pgl-rank-five-support}
\label{thm:pgl-minimal-self-reconstructing-boundary-carrier}
Equations~\eqref{eq:pgl-rank-five-carrier}-\eqref{eq:pgl-source-generated-finite-core} are the least marked support and generated fixed-center algebra on the primitive branch. The ordered section-ring product supplies the unique $\mu:\operatorname{Sym}^2W\to C$; a marked line $L\subset W$ determines $L^{\otimes2}\subset C$, so the degree-one line, degree-two line, oriented axis, and Veronese source form one typed moment tower. The homogeneous split twist is $L_\Gamma^3\oplus L_\Gamma^2$, with determinant $L_\Gamma^5$ and standard twisted chiral kernel $E_3\oplus E_2$; the weight difference one excludes an $SU(2)$-equivariant zeroth-order chiral mixing.
\end{theorem}
The corresponding exceptional/Kronecker description is standard on $\mathbb P^1$ \cite{Beilinson1978CoherentSheaves}; its Tits and stability readings are represented by \cite{Kac1980InfiniteRootSystems} and \cite{King1994Moduli}.

\begin{proposition}[Primitive port projector packet]
\label{prop:pgl-primitive-exact-sequence-self-dilation}
\label{prop:pgl-primitive-port-projector-packet}
The multiplication globalizes to
\begin{equation}
0\longrightarrow L_\Gamma^{-1}\otimes W\xrightarrow{\ \jmath\ }\mathcal O\otimes C\longrightarrow L_\Gamma^2\longrightarrow0,
\label{eq:pgl-primitive-multiplication-exact-sequence}
\end{equation}
and, after dualization and tensoring by $L_\Gamma$, to $0\to L_\Gamma^{-1}\to L_\Gamma^{\oplus3}\to(L_\Gamma^2)^{\oplus2}\to0$. For unit $s\in W$, $P_s=|s\rangle\langle s|$ and $X=\sqrt\alpha\,\jmath_s$, the nonzero squared singular values are $\{\alpha,\alpha/2\}$. With $K=XX^\dagger$,
\begin{equation}
P_D={\bf1}_C-X(X^\dagger X)^{-1}X^\dagger={\bf1}_C-KK^+,
\qquad
L_{\rm dark}=\operatorname{Ran}P_D=\operatorname{Sym}^2(s^\perp),
\label{eq:pgl-primitive-dark-projector}
\end{equation}
For the tautological $S\simeq\mathcal O(-1)$ and quotient $Q_W=\underline W/S\simeq\mathcal O(1)$, this gives $L_{\rm dark}\simeq\operatorname{Sym}^2Q_W\simeq Q_W^{\otimes2}\simeq L_\Gamma^2$. The rank-two port is one-step cyclic, so its compressed Hermitian resolvent has minimal realization dimension five.
\end{proposition}

\begin{corollary}[Rigid primitive exact complex]
\label{cor:pgl-rigid-primitive-chain}
Viewed as the three-vertex exact holomorphic complex in \eqref{eq:pgl-primitive-multiplication-exact-sequence}, the canonical polynomial frame has map-variation dimensions $12+9$, linearized-relation rank $8$, and vertex-gauge orbit dimension $4+9+1-1=13$. Hence $H^1_{\rm def}=0$ and $H^0_{\rm def}\simeq\mathbb C$, in agreement with the King-stable Schur description \cite{King1994Moduli}. Its underlying two-arrow holomorphic chain is used below for the stability chamber. The remaining realization freedom lies in Hermitian metrics, connections, common coupling/area, and radial embedding.
\end{corollary}

\begin{proposition}[Dark Veronese projector geometry]
\label{prop:pgl-dark-veronese-projector-geometry}
For $g(P)=\Tr(dP\,dP)/2$ and $\Omega(P)=-i\Tr(P\,dP\wedge dP)$, let $P_Q$ project onto $S^\perp$ and $P_D$ be \eqref{eq:pgl-primitive-dark-projector}. Then
\begin{equation}
g_D=2g_Q,
\qquad
\Omega_D=2\Omega_Q,
\qquad
\Tr(P_D(s)P_D(t))=\left[\Tr(P_Q(s)P_Q(t))\right]^2.
\label{eq:pgl-dark-veronese-projector-geometry}
\end{equation}
The image is the degree-two Veronese conic; it saturates the metric-curvature inequality and, for $c_1(L_{\rm dark})=2$, has QGT area $2\pi$.
\end{proposition}
The complex dark conic and the real source Veronese orbit are distinct typed readings of the same symmetric-square representation.

\begin{equation}
\sum_iA_iA_i^\dagger={\bf1}_C,
\qquad
\sum_iA_i^\dagger A_i=\frac32{\bf1}_W,
\qquad
\mu_K=P_C-\frac32P_W=\frac12H_\Gamma^{\rm deg}.
\label{eq:pgl-kronecker-moment-frame}
\end{equation}
The associated completely positive map is the adjacent-spin covariant channel \cite{AschieriRubaSolovej2024SU2Channels}, equivalently the symmetric cloning endpoint \cite{Werner1998OptimalCloning}. Exact coherent-symbol recovery fixes $\mathcal L_K=\frac13P_{V_1}+P_{V_2}$ and
\begin{equation}
r_K=\sqrt{\frac23},
\qquad
G_K:=\mathcal L_K^+=3P_{V_1}+P_{V_2}.
\label{eq:pgl-kronecker-green-metric}
\end{equation}
The balanced Kempf-Ness quotient lies on this positive ray. Finally,
\begin{equation}
\mathsf J_H=\lambda_H+\rho_H,
\qquad
\mathsf K_H=i(\lambda_H-\rho_H),
\qquad
\lambda_H=\frac12(\mathsf J_H-i\mathsf K_H),
\label{eq:pgl-standard-left-right-jordan-lie}
\end{equation}
fixes the Jordan/Lie normalization used in the Parent response. These are dependent readings of one primitive multiplication packet rather than separate reconstruction mechanisms.

\subsection{Gauge form, chiral packet and hypercharge normalization}

\begin{proposition}[Centered marked grading calculus]
\label{prop:pgl-centered-marked-grading}
\label{prop:pgl-two-block-degree-calculus}
For a marked Hermitian carrier $V=\bigoplus_jV_j$ with ranks $d_j$ and distinct integral levels $\ell_j$, center $r\ell_j-\sum_kd_k\ell_k$ and divide by the common divisor. For a two-block split of levels $(2,1)$ and ranks $c>w\geq2$, with $g=\gcd(c,w)$,
\begin{equation}
H_{c,w}^{\rm prim}=\frac{w}{g}P_C-\frac{c}{g}P_W,
\qquad
\lambda_{c,w}^{\rm prim}(U(1))\cap(SU(c)\times SU(w))\simeq\mu_{cw/g}.
\label{eq:pgl-general-primitive-cocharacter}
\end{equation}
On the coprime branch the rational degree has weights $-1/c,1/w$, $\Tr_V(H^2)=cw(c+w)$, $\dim\Lambda^{\rm even}V=2^{c+w-1}$, and Abelian/non-Abelian trace ratio $2(c+w)/(cw)$.
\end{proposition}
For $(c,w)=(3,2)$ the retained volume morphism gives
\begin{equation}
S(U(3)\times U(2))\simeq\frac{SU(3)_c\times SU(2)_L\times U(1)_Y}{\mathbb Z_6}.
\label{eq:pgl-standard-model-global-form}
\end{equation}
The same tensor automorphism group follows from Doplicher-Roberts reconstruction \cite{DoplicherRoberts1990}. Line/cocharacter comparisons are represented by \cite{AharonySeibergTachikawa2013LineOperators}, \cite{Tong2017LineOperators}, and \cite{GaiottoKapustinSeibergWillett2015Generalized}; filtered-fiber conventions follow \cite{Ziegler2015FilteredFiberFunctors}.

\begin{proposition}[Determinant-dual Clifford and weak-direct calculus]
\label{prop:pgl-reduced-clifford-completion}
\label{prop:pgl-determinant-dual-weak-direct}
On $F_\Gamma=\Lambda^\bullet V$, exterior multiplication and contraction satisfy the canonical CAR \cite{AtiyahBottShapiro1964CliffordModules}. Put $F_\Gamma^+=\Lambda^{\rm even}V$ and $F_\Gamma^-=\Lambda^{\rm odd}V$, with $\dim F_\Gamma^\pm=16$. The $|1|$-grading of $\mathfrak{sl}(V)$ has degree-zero Levi $C\oplus W$ and positive part $\operatorname{Hom}(W,C)$; the primitive multiplication selects its protected $E_2$ slice. With the marked volume duality and $h\in W$, the four multiplicity-one color-singlet weak-direct blocks are
\begin{equation}
\begin{aligned}
T_\nu(h)&=\varepsilon(h)P_{\Lambda^0V},&
T_u(h)&=\varepsilon(h)P_{\Lambda^2C},\\
T_e(h)&=\iota(h^*)P_{\Lambda^2W},&
T_d(h)&=\iota(h^*)P_{\Lambda^2C\otimes\Lambda^2W}.
\end{aligned}
\label{eq:pgl-weak-direct-odd-blocks}
\end{equation}
They are odd for the exterior grading; their $B-L$ property follows from the derived operator in \eqref{eq:pgl-derived-bl-envelope}. A spacetime fermionic-field interpretation is represented by \cite{Fewster2016SpinStatisticsCurved} and is not used to determine the finite packet. For the marked unit volume $\Omega_W\in\Lambda^2W$, the canonical antiunitary $j_W(h):=\iota(h^*)\Omega_W$ satisfies $j_W^2=-{\bf1}$ and pairs orthogonal weak lines. With $D_{\det}(\xi):=\xi\wedge\Omega_W$, one has $T_d(h)D_{\det}=T_u(j_Wh)$, so the determinant/Hodge exchange carries neither an independent family phase nor an additional sector gain.
\end{proposition}
The induced Real/Hodge bookkeeping is compared with finite NCG in \cite{ConnesLott1991ParticleModelsNCG}; spacetime spin-statistics is represented by \cite{Fewster2016SpinStatisticsCurved}. These comparisons do not determine the microscopic weak-direct block.

\begin{equation}
Y=d\Gamma\!\left(-\frac16H_\Gamma^{\rm deg}\right)=-\frac13N_C+\frac12N_W,
\label{eq:pgl-hypercharge-degree}
\end{equation}
The split has the standard AIII envelope $\operatorname{Gr}_2(\mathbb C^5)\simeq SU(5)/S(U(3)\times U(2))$, and $F_\Gamma^+\simeq{\bf1}\oplus{\bf10}\oplus\bar{\bf5}$ is the $SU(5)$ restriction of the chiral Spin$(10)$ spinor \cite{BaezHuerta2010GUTAlgebra}. For $X_{\rm deg}=5-2(N_C+N_W)$,
\begin{equation} B-L=\frac{X_{\rm deg}+4Y}{5}. \label{eq:pgl-derived-bl-envelope} \end{equation}
The weak-direct blocks preserve \eqref{eq:pgl-derived-bl-envelope} because $\Delta N=\pm1$ and $\Delta Y=\pm1/2$. This is a representation-theoretic envelope, not a physical GUT assumption.
The exterior character gives $I_3=I_2=2$, $\Tr_{F_\Gamma^+}(Y^2)=10/3$ and $k_Y=5/3$. The protected gauge trace form is
\begin{equation}
B_F(X,Y):=-\operatorname{Tr}_{F_\Gamma^+}\!\left(\rho_*(X)\rho_*(Y)\right).
\label{eq:pgl-protected-gauge-trace-metric}
\end{equation}
Its ray is structural; an absolute representative remains a Parent-normalization question.

\begin{proposition}[Chirality-independent anomaly-selected central cocharacter]
\label{prop:pgl-top-form-determinant-kernel}
For $Q_{a,b}=aN_C+bN_W$ on either chiral half,
\begin{align}
\operatorname{Tr}Q_{a,b}&=8(3a+2b),&
\mathcal A_{SU(3)^2Q}&=\mathcal A_{SU(2)^2Q}=3a+2b,\notag\\
\operatorname{Tr}Q_{a,b}^3&=4(3a+2b)(9a^2+6ab+5b^2).
\label{eq:pgl-central-anomaly-factorization}
\end{align}
Hence a nontrivial central gauged $U(1)$ is locally anomaly-free exactly on $3a+2b=0$, with primitive generator $(-2,3)$ and
\begin{equation}
\mathfrak g_\Gamma=s(u(C)\oplus u(W)).
\label{eq:pgl-local-anomaly-kernel}
\end{equation}
The same exterior character gives the vanishing Standard-Model local anomalies and the $\mathbb Z_6$ kernel in \eqref{eq:pgl-standard-model-global-form}; the one-block $I(J)=\{2\}$ alternative retains a nonzero cubic $SU(3)^3$ coefficient. For the exact global group $G=S(U(3)\times U(2))$ in \eqref{eq:pgl-standard-model-global-form}, the spin-bordism computation of \cite{DavighiGripaiosLohitsiri2020GlobalAnomalies} gives $\Omega_5^{\rm Spin}(BG)=0$. Thus, on spin four-manifolds in the Dai-Freed setting, the spin-bordism group contributes no additional global gauge-anomaly obstruction after the local anomaly conditions have passed; gauging an additional finite family symmetry would require a new bordism audit.
\end{proposition}
A connection-compatible determinant-line trivialization is treated with \cite{BismutFreed1986EllipticFamiliesII}; characteristic and residual flat/torsion conventions follow \cite{MilnorStasheff1974} and \cite{ChernSimons1974}.

\subsection{Equianharmonic family structure and the norm-seven sector}
\label{subsec:pgl-equianharmonic-family}

The arithmetic reading uses the integral $A_2$ package in three steps. The BNR line selects the self-dual spectral branch, the marked $A_2/E[2]$ orbit supplies the common $C_3$ family cycle, and the norm-seven theta lift produces the conjugate triplets and their structural orientation; physical orientation remains the later Wilson gate. Standard theta multiplicities are represented by \cite{Mumford1966AbelianVarietiesI}; tubular and root-stack comparisons by \cite{ChenChen2017Tubular} and \cite{Cadman2007UsingStacks}. Weighted-projective and tubular representation comparisons are represented by \cite{ChenKrause2009WeightedProjectiveLines}, with the quantum-loop comparison in \cite{Schiffmann2004QuantumLoopAlgebras}. On the primitive parity branch they give the same rank order $(3,2)$.

\begin{proposition}[Self-dual spectral splitting]
\label{prop:pgl-bnr-cohomological-splitting}
Let $D$ be the marked trinion divisor, $L=K_{\mathbb P^1}(D)\simeq\mathcal O(1)\simeq L_\Gamma$, and let a degree-zero rank-three $L$-twisted Higgs bundle have smooth cyclic spectral curve $q:E\to\mathbb P^1$. The BNR spectral line $N$ is defined by the standard cokernel construction \cite{BeauvilleNarasimhanRamanan1989Spectral} and satisfies $q_*N\simeq\mathscr H$. Put $N_{{\rm sp},0}=q^*L_\Gamma$ and $P_{\rm sp}=N\otimes N_{{\rm sp},0}^{-1}$. A nondegenerate symmetric form with $\Phi^{\mathsf T}B_{\rm sp}=B_{\rm sp}\Phi$ gives $P_{\rm sp}^2\simeq\mathcal O_E$, hence $P_{\rm sp}\in E[2]$, and
\begin{equation}
\mathscr H\simeq
\begin{cases}
L_\Gamma\oplus\mathcal O\oplus L_\Gamma^{-1},&P_{\rm sp}=0,\\
\mathcal O^{\oplus3},&P_{\rm sp}\in E[2]^\times.
\end{cases}
\label{eq:pgl-bnr-cohomological-family-selector}
\end{equation}
\end{proposition}

\begin{proposition}[Root-line family orbit]
\label{prop:pgl-finite-core-full-shadow-representation}
\label{prop:pgl-parity-character-corner}
The root lines satisfy $\Phi(A_2)/\{\pm1\}\simeq(\Lambda_E/2\Lambda_E)^\times\simeq E[2]^\times$. The nonzero orbit has the common degree-zero bundle type of Proposition~\ref{prop:pgl-bnr-cohomological-splitting}, with simply transitive Coxeter action. For the oriented regular cycle $S:=G_{\rm fam}$ the marked eigenframe fixes
\begin{equation}
R_{\rm cyc}:=\operatorname{diag}\!\left(0,-\frac13,\frac13\right),\qquad
S=e^{-2\pi iR_{\rm cyc}},\qquad H_F=-3R_{\rm cyc},\qquad K_F=9R_{\rm cyc}^2-\frac23{\bf1}.
\label{eq:pgl-family-cycle-operators}
\end{equation}
Thus $H_F^3=H_F$ and these are Fourier readings of one $C_3$ permutation. Its microscopic interpretation requires the common marked connection; mass filtration remains a later attachment. Parity induction gives the two-arrow Kronecker heart, while the integral regular family lattice is the nonsplit extension of its invariant line by the $A_2$ augmentation lattice.
\end{proposition}

Graph/bundle-Laplacian conventions are represented by \cite{FriedmanTillich2004GraphCalculus} and \cite{Kenyon2011VectorBundleLaplacian}; finite spectral-module comparisons by \cite{ChristensenIvan2006AFSpectralTriples}. The family tower is generated by
\begin{equation}
J_q={\bf1}-\frac{9q}{2}R_{\rm cyc}-\frac{27}{2}R_{\rm cyc}^2=a_qS+\bar a_qS^\dagger,
\qquad \Theta(q)=-\frac{a_q^2}{N(q)}.
\label{eq:pgl-family-fourier-master}
\end{equation}
Positive Gram data retain its Wilson magnitude, while conjugation remains oriented pre-Gram data.

\begin{proposition}[Cartan-Eisenstein odd-tower master]
\label{prop:pgl-primitive-commensurate-torsor-spectrum}
For $q=2m+1$ put $K_m=\operatorname{diag}(0,m,2m+1)$ and $K_{0,m}=K_m-(3m+1){\bf1}/3$. In the marked spectral order of \eqref{eq:pgl-family-fourier-master}, $-3K_{0,m}$ has the signed spectrum of $J_q$, with $N_m=3m^2+3m+1$, $27\det K_{0,m}=3N_m-1$, $q^2=(4N_m-1)/3$, and $27\operatorname{Disc}\chi_{K_{0,m}}=(N_m-1)^2(4N_m-1)$. The same tower is represented by $\xi_m=1+m(1-\omega)$, $N_E(\xi_m)=N_m$ and torsor triple $(1,3m+1,3m+2)$. Primitive commensurability gives the first positive simple class $N=7$ at $m=1$, equivalently $q=3$, root gaps $(1,2,3)$ and torsor triple $(1,4,5)$.
\end{proposition}

\begin{proposition}[CM triality and canonical norm-seven orientation packet]
\label{prop:pgl-cm-seven-isogeny-packet}
For $\xi=2-\omega$, multiplication on $E_\omega$ has determinant $7$, $[\bar\xi]\circ[\xi]=[7]$, and pulls $\Theta=\mathcal O_E(O)$ back to $L_7=\mathcal O_E(7O)$. Since $\xi\equiv\omega\pmod2$, the induced $C_3$ acts on $E[2]^\times$, while modulo $\xi$ the nonzero characters split as $Q=\{1,2,4\}$ and $NQ=\{3,5,6\}$. For the canonical symmetric linearization of $L_7$, $H^0(E,L_7)\simeq\mathbb C[\mathbb F_7]\simeq{\bf1}\oplus V_6$ as a $G_{42}=AGL(1,7)$ module, with $V_6|_{F_{21}}\simeq3\oplus\bar3$ for $F_{21}=C_7\rtimes C_3$. Let $\mathbf e=(1,\ldots,1)^{\mathsf T}$, $P_1=\mathbf e\mathbf e^{\mathsf T}/7$, and $P_{\rm aug}=I-P_1$. For the Paley adjacency $A_{Q,xy}=1$ iff $y-x\in Q$, put $J_7=(A_Q-A_Q^{\mathsf T})/\sqrt7$ and $D_7=-iJ_7$. Then $\Aut(T_7)=F_{21}$, $\Aut^{\pm}(T_7)=G_{42}$, $J_7^2=-P_{\rm aug}$, and
\begin{equation} P_3=\frac12(P_{\rm aug}+D_7),\qquad P_{\bar3}=\frac12(P_{\rm aug}-D_7). \label{eq:pgl-paley-theta-packet} \end{equation}
The $G_{42}$- and $F_{21}$-commutants have dimensions $2$ and $3$. Thus an $F_{21}$-equivariant Hermitian structural operator is uniquely $H=h_1P_1+h_0P_{\rm aug}+\eta D_7$, with $h_1,h_0,\eta\in\mathbb R$ and full $G_{42}$ invariance iff $\eta=0$; cameral sheet exchange sends $\eta\mapsto-\eta$ and exchanges $P_3,P_{\bar3}$. This is a statement about the structural theta carrier before microscopic mass filtration.
\end{proposition}
The marked Eisenstein convention fixes the norm-seven associate up to conjugation: $(1+\omega)(2-\omega)=3+2\omega$ and $e^{i\delta_7}=-[(3+2\omega)/\sqrt7]^6=-(\xi/\bar\xi)^3$. With $e^{i\theta_1}=-\omega^2\xi/\bar\xi$, equivalently $e^{i\delta_7}=(e^{i\theta_1})^3$, so the phase is a dependent first-cell reading. Physical typing remains the determinant/holonomy incidence below. Elliptic-surface comparisons are represented by \cite{SchuettShioda2010EllipticSurfaces} and \cite{OguisoShioda1991MordellWeil}. In the Real $E_\omega[2]\simeq\mathbb F_2^2$ frame, multiplication by $\omega$, conjugation and translations generate the marked $C_3\rtimes C_2$ and $\operatorname{Aff}(E_\omega[2])\simeq S_4$; the nonnormal Klein and TBM overlap are dependent reference data. Recent algebraic family realizations are represented by \cite{GourlayGresnigt2024}, with a broader survey in \cite{Furey2025}.

\begin{theorem}[Primitive protected finite geometry]
\label{thm:pgl-primitive-graded-standard-model-reconstruction}
On the primitive two-grade branch, Gauss-local realization and PLSR determine the principal $A_2$ source, the carrier \eqref{eq:pgl-rank-five-carrier}, and the algebra \eqref{eq:pgl-source-generated-finite-core}. Their protected descendants are generated by the primitive exact complex of Proposition~\ref{prop:pgl-primitive-port-projector-packet}, its determinant/exterior functor, and the integral $A_2$ datum. These yield respectively the port/dark/Green geometry; the global gauge and chiral packet including $k_Y$; and the family cycle \eqref{eq:pgl-family-cycle-operators}, Cartan-Eisenstein tower of Proposition~\ref{prop:pgl-primitive-commensurate-torsor-spectrum}, and norm-seven packet of Proposition~\ref{prop:pgl-cm-seven-isogeny-packet}. Their union is denoted by $\mathfrak P_\Gamma$.
\end{theorem}
The inverse checks do not enter the forward PLSR proof. Projector/Gram data precede phase-sensitive transport, and a physical spectral state still requires the microscopic attachment of Subsection~\ref{subsec:alena-microscopic-normal-family}.

\section{A Constrained Microscopic Realization}
\label{sec:alena-codazzi-realization}

Section~\ref{sec:primitive-gauss-local-self-reconstruction} fixes $\mathfrak P_\Gamma$ before microscopic dynamics. Non-emptiness of the declared realization class is established by constructing one same-domain Alena-Codazzi Parent; uniqueness is not claimed. Constraint reduction is performed first, the intrinsic terminal is then identified, and microscopic transport and response gates are attached only afterwards. The construction is frozen before downstream holdouts are evaluated, with independent incidences collected in $\mathfrak C(\mathscr L_{\rm Parent})$ of \eqref{eq:fsr-strict-saturation-stratum}.

\begin{proposition}[Constraint-reduced realization calculus]
\label{prop:fsr-constraint-reduced-realization}
Let a regular local Parent constraint be written as $\mathcal C(m,r)=0$, with $C_M:=D_m\mathcal C$ invertible and $C_R:=D_r\mathcal C$. The physical tangent is the graph of $A_C=-C_M^{-1}C_R$. With
\begin{equation}
G_C=\begin{pmatrix}A_C\\{\bf1}\end{pmatrix},\qquad
W_C=G_C^\dagger G_C,\qquad
U_C=G_CW_C^{-1/2},
\label{eq:fsr-constraint-physical-frame}
\end{equation}
$U_C$ is an orthonormal frame; intrinsically the constrained-tangent projector is $U_CU_C^\dagger={\bf1}-D\mathcal C^\dagger(D\mathcal C D\mathcal C^\dagger)^{-1}D\mathcal C$. If $L=(L_M\ L_R)$ is the derivative of a native target map and the frozen representative is stationary for the Parent functional $\mathscr S$, let $H=D^2\mathscr S|_*$ be its raw Euler Hessian on the same domain. The physical representatives are
\begin{equation}
\mathscr L_C=(L_R+L_MA_C)W_C^{-1/2},\qquad
K_C=U_C^\dagger HU_C.
\label{eq:fsr-constraint-native-pullback}
\end{equation}
At a merely constrained critical point, $H$ is replaced by the Hessian of the corresponding augmented functional. Constraints are therefore applied before positivity, fixed-charge reduction, or spectral flattening. If $\mathscr L_C$ is surjective, local target lifting has no cokernel obstruction; a target component obtained by projection of the same $\mathscr L_C$ and its adjoint pullback are derived readings rather than additional incidences. Positive spectral refinements may be removed only after the protected kernel, support, and open gap have been fixed.
\end{proposition}

\subsection{Parent action, Codazzi constraints and resolved source}
\label{subsec:alena-primitive-realization-class}

The local Parent is kept in the single-density Alena form
\begin{equation}
S_A[\Phi]=\sigma_A\int_X p_\Lambda\bigl[1-\varepsilon_A(\rho_\zeta,R_\omega)\bigr] \,d{\rm vol}_\eta,
\qquad \sigma_A=\pm1,
\label{eq:alena-one-function-action}
\end{equation}
with the off-shell native map $\Psi\mapsto(J,U,\kappa_{\rm base},R_{\rm vort})$. For $p_0>0$ and a positive invariant gauge form $G_0$,
\begin{equation}
\kappa^2=\kappa_{\rm base}^2-\frac{1}{4p_0}G_0(\mathcal F_{\mu\nu},\mathcal F^{\mu\nu}),\qquad
p_\Lambda=p_0\kappa^2,\qquad
\eta^{-1}k=P_E+\left(\frac\kappa2-1\right)P_F,
\label{eq:alena-parent-metric-density-normal-form}
\end{equation}
The gauge metric is structurally fixed on the ray of \eqref{eq:pgl-protected-gauge-trace-metric}; dynamical-principal-bundle comparison is represented by \cite{Helein2025DynamicalPrincipalBundle}. On the regular scalar branch the same-response condition fixes uniquely $\mu=-\varphi'$ and $V=1-\varphi+\rho\varphi'$. For $\varphi''>0$ the active point is the unique minimum, and the cyclic phase gives
\begin{equation}
J^\mu_{\rm tr}=p_\Lambda\zeta^2U^\mu,
\qquad
\nabla^{(k)}_\mu J^\mu_{\rm tr}=0.
\label{eq:alena-realization-current}
\end{equation}
The stress/current conventions follow \cite{AT1} and \cite{ogonowski2025halo}; continuum and branch inputs are represented by \cite{AT2}, \cite{AT3}, and \cite{ATNew}. The reduced perspective law is
\begin{equation}
\mathscr W(t)=1-\varphi(t),\qquad
\mathcal L_{\rm red}=\sigma_Ap_\Lambda[1-\alpha+\alpha\varphi(t)],\qquad
\mathscr W'(t)=\mu(t),\quad \mathscr W''(t)=-\varphi''(t).
\label{eq:alena-parent-perspective-response}
\end{equation}
so the scalar susceptibility, Legendre metric and rank-one phase/scalar Hessian are descendants of one function $\varphi$. The full perspective Hessian is
\begin{equation}
D^2_{(R_{\rm vort},\nu_\vartheta)}(\sigma_A^{-1}\mathcal L_{\rm red})
=\frac{p_\Lambda\varphi''(t)}{\alpha}
\begin{pmatrix}1\\d_kt\end{pmatrix}
\begin{pmatrix}1&d_kt\end{pmatrix}.
\label{eq:alena-parent-perspective-rank-one}
\end{equation}
In particular,
\begin{equation}
g_{\rm sc}=\mathcal V_{\mu\mu}\,d\mu^2=C_R\,dR^2,\qquad \nabla^2(\Phi\circ R)=C_R\,dR\otimes dR-\mu\,\nabla^2R.
\label{eq:alena-scalar-legendre-hessian}
\end{equation}

\begin{proposition}[Regular Codazzi normal form and flat one-sided reduction]
\label{prop:alena-codazzi-adapted-principal-geometry}
On a regular two-eigenvalue multiplier patch, the standard Codazzi identities \cite{DerdzinskiShen1983Codazzi} and their Lorentzian extension \cite{ManticaMolinari2023CodazziSpaceTimes} give integrable umbilic eigendistributions and the local conformal-product form
\begin{equation}
k=(\lambda_E-\lambda_F)^{-2}\bigl(h_E(x)\oplus h_F(y)\bigr),
\qquad \lambda_E=\lambda_E(y),\quad\lambda_F=\lambda_F(x).
\label{eq:alena-codazzi-conformal-product}
\end{equation}
In the retained Parent normal form $u_{\rm AC}=3(\lambda_E+\lambda_F)/(\lambda_F-\lambda_E)$, with $\lambda_E<0<\lambda_F$ and $-1<\lambda_E/\lambda_F<-1/2$. Flatness forces $d\lambda_E\equiv0$ or $d\lambda_F\equiv0$ on each connected regular patch, so the nonparallel branches are dual affine warps. Writing $S_B=b_EP_E+b_FP_F$, the standard eigenbundle formula gives
\begin{equation}
\theta^B=-P_E^*d\log|b_E|-P_F^*d\log|b_F|,
\qquad
H_E=\frac{P_F((db_E)^\sharp+b_E(\theta^B)^\sharp)}{b_E-b_F},\quad
H_F=\frac{P_E((db_F)^\sharp+b_F(\theta^B)^\sharp)}{b_F-b_E}.
\label{eq:alena-codazzi-spectral-slope}
\end{equation}
On the multiplier locus $\theta^B=d\log|\varphi|$. A constant nontrivial ratio $b_E/b_F$ makes $J$ and $A$ parallel; the regular nonparallel sector therefore requires spectral-shape variation. The one-sided family below realizes this nonlinear Codazzi form explicitly. Perturbative Codazzi equations are subsequently treated as domain constraints, so their reaction directions are eliminated before the physical Hessian is formed.
\end{proposition}
The trace-free angular Codazzi symbol has kernel $V_0\oplus V_1$ and first high-angular squared singular value $12$, so the constraint is applied once to define the physical graph before any Hessian or selector test. The same principal symbol fixes the source typing used below: for $\zeta=\tau U^\flat+\xi$ its trace-free symmetric kernel is nonzero exactly for null $\zeta$, generated by $\zeta\otimes\zeta$. The spatial trace-free projection is the Veronese class $Q_e$ of \eqref{eq:pgl-veronese-source-orbit}, while the mixed component obeys $\tau\pi=(\xi\odot q)_0$ and supplies the multiplicity-one Jordan tensor. Parent normalization of this incidence remains independent.

\begin{proposition}[Constructive one-sided Parent seed geometry]
\label{prop:alena-constructive-parent-nonempty}
On the branch $\lambda_E=-a$, $a>0$, one has
\begin{equation} \varphi_E(d)=\frac{3ad}{4p_0(d-1)(d+1)(d+2)},\qquad \lambda_F=a\frac{2d+1}{d+2},\qquad \Delta_E=-3a\frac{d+1}{d+2}. \label{eq:alena-affine-warp-seed-coefficients} \end{equation}
Put $g_E:=dr^2+\epsilon\,dy^2$. For $\eta=g_E+r^2\bar h$, $K_{\bar h}=1$, and
\begin{equation}
r(d)=\frac{c}{3a}\frac{d+2}{(d+1)\sqrt d},
\qquad
\Delta_E^2d\,r(d)^2=c^2,
\qquad
\frac{dr}{dd}<0.
\label{eq:alena-affine-warp-seed-representative}
\end{equation}
the metric $k=g_E+d(r)r^2\bar h$ realizes \eqref{eq:alena-codazzi-conformal-product} on a genuinely nonparallel local branch. Choosing the unit field $U$ along the translational $y$ direction and all scalar data $y$-independent satisfies the current equation \eqref{eq:alena-realization-current}. The same profile fixes the active scalar law uniquely by $\mu=-\varphi_E'$ and $V=1-\varphi_E+d\varphi_E'$. On $d>1$ one has $\mu>0$ and $C_R=\varphi_E''>0$; the condition $C_R(d_*)=40/49$ fixes the amplitude used by the global selector below and remains a cross-layer incidence. The constant parallel/product configurations are recovered as the vanishing-first-jet sublocus. A compact hyperbolic resolved profile remains a regulated scalar/Green benchmark under the warped-product and elliptic conventions of \cite{Besse2007} and \cite{GilbargTrudinger2001Elliptic}.
\end{proposition}
The formal $d\to\infty$ continuation is not used as a core regularization. With $C=c/(3a)$, it has $r\sim Cd^{-1/2}$ and a finite-area cylindrical transverse limit, while $p_\Lambda V\sim4p_0C^4r^{-4}$ and the positive phase flux through a cutoff collar grows as $(16\pi/3)p_0C^6\epsilon^{-3}$. The proposition therefore supplies a local/exterior nonparallel branch; a finite charged core requires an independently regularized completion.
Along the one-sided affine-warp seed the zero-derivative density already has the fixed slope
\begin{equation}
\frac{d}{dd}\bigl[p_\Lambda\varphi_E(d)\bigr]=-6a\frac{2d+1}{(d-1)^2(d+2)^2}<0.
\label{eq:alena-affine-warp-ultralocal-no-go}
\end{equation}
Stationarity of the nonparallel branch therefore requires the compensating native first force; it cannot be supplied by an ultralocal scalar term alone. The same profile gives a Hilbert-compatible collar anchored at zero radial flux $\mathfrak c=0$ and $\Omega=h=0$. At fixed collar Noether charge $Q_I$, its even branch has $R_H=d-[\varphi_E/\mu]h_2\mathfrak c^2+O(\mathfrak c^4)$ with $h_2>0$, hence $R_2<0$ without a fitted second jet. This collar is retained as a same-law holdout rather than used to define the global selector. Compact-core/defect regularization conventions are represented by \cite{FedererFleming1960IntegralCurrents}, \cite{Simon1983GeometricMeasureTheory}, \cite{BethuelBrezisHelein1994GinzburgLandau}, and \cite{JerrardSoner2002GinzburgLandau}.

After density elimination, the native first and second action jets follow directly by differentiating $\sigma_Ap_\Lambda\widehat\varphi(R_{\rm vort})$; the zero-derivative affine-warp coefficient is fixed by the same seed and is not an independent fit. On the reference phase $\Xi_g=\varphi_E(R)$, and the gauge-flat Maxwell second variation is
\begin{equation}
D^2_{\mathcal A}S_A[a,a]=-\frac{\sigma_A\varphi_E(R)}{2}\int_XG_0(\mathfrak f,\mathfrak f)\,d\mathrm{vol}_\eta.
\label{eq:alena-native-maxwell-second-jet}
\end{equation}
Writing $c_{\rm r}$ for the radial integration constant $c$ in \eqref{eq:alena-affine-warp-seed-representative}, the classical Parent has the exact positive scale orbit $(p_0,G_0,a,c_{\rm r})\mapsto(\lambda p_0,\lambda G_0,\lambda a,\lambda c_{\rm r})$, $\lambda>0$. Dimensionless shape data, the selector/fold and normalized response chain are invariant, whereas $S_A$, the Noether charge and quadratic forms derived from the same action at fixed field coordinates have weight one. Thus scale-free data cannot determine the absolute Parent normalization. If $G_0=\gamma B_F$ and $m:=(2\pi)^{-1}\int_F\star_kJ_{\rm tr}^{\flat}$ is the marked period over a transverse three-cell $F$, both $\gamma$ and $m$ have weight one; $m/\gamma$ is therefore the scale-free current/gauge collision coordinate, leaving one common normalization scale for this pair. No value of $m$ is supplied by the present construction. If the protected gauge ray is independently lifted to $G_0=B_F$, take unit marked perturbations $G_0(\mathfrak f,\mathfrak f)=1$ and $(\delta R)^2=1$ and put $h_g:=-\sigma_A\varphi_E(R)/2$, $h_R:=\sigma_A4p_0(d+1)^2C_R(R)$. The same-action pair then yields
\begin{equation} p_0=-\frac{\varphi_E(R)}{8(d+1)^2C_R(R)}\frac{h_R}{h_g}. \label{eq:alena-parent-absolute-scale-bridge}\end{equation}
This bridge remains conditional on the absolute marked gauge normalization and on an independent primitive-action normalization.

The resolved quadratic source retains
\begin{equation}
0\longrightarrow\mathbb R\longrightarrow V_2\oplus V_1\xrightarrow{\ \Phi_{\xi,B_0}\ }\mathcal S_{\xi,B_0}\longrightarrow0,
\qquad
0\longrightarrow\mathbb RB_0\longrightarrow V_2\longrightarrow\mathcal S_{\xi,B_0}/\operatorname{Im}\mathscr M_{B_0}\longrightarrow0,
\label{eq:alena-source-hodge-sequence}
\end{equation}
\begin{theorem}[Veronese Frenet-Gauss source jets]
\label{thm:alena-dual-veronese-source-jet}
For unit-speed $e(t)\subset S^2$ with signed geodesic curvature $\kappa_g$, write $Q(e(t))=Q_e+tA+t^2B+O(t^3)$, $B_H=B+3Q_e/2$, $\chi_H=\sqrt3\|B_H\|/\|A\|^2$, and $\chi_R=\sqrt3\|B\|/\|A\|^2$. The standard Veronese Frenet identities give
\begin{equation}
\chi_H^2=\frac38(1+\kappa_g^2),
\qquad
\chi_R^2=\frac38(4+\kappa_g^2),
\qquad
\chi_R^2-\chi_H^2=\frac98.
\label{eq:alena-dual-veronese-quotients}
\end{equation}
The third jet is fixed by the same Frenet data; on the geodesic branch $Q^{(3)}=-4A$, so no independent cubic source direction is present.
\end{theorem}

\subsection{Constrained tangent and marked-residue response}
\label{subsec:alena-native-principal-root}

The regular high-angular domain is reduced by Codazzi before positivity or branch release. Proposition~\ref{prop:fsr-constraint-reduced-realization} therefore applies directly: the physical tangent is the graph of $A_C=-C_M^{-1}C_R^{\mathcal C}$, with $\|C_M^{-1}\|\le1/\sqrt{12}$ from the angular gap. Native first and second jets are pulled back once to this graph.

\begin{proposition}[Codazzi-conditioned native submersion and graph stability]
\label{prop:alena-codazzi-native-submersion}
Let $s_R:=\sigma_{\min}^{\rm row}(L_R)>0$ and $\ell_M:=\|L_M\|$. Then
\begin{equation}
s_C:=\sigma_{\min}^{\rm row}(\mathscr L_C)
\ge s_{\rm sub}:=
\frac{s_R-\ell_M\alpha_C}{\sqrt{1+\alpha_C^2}}.
\label{eq:alena-codazzi-submersion-margin}
\end{equation}
The native $U$ component is taken on the unit target $\mathcal U_k:=\{u:k(u,u)=\varepsilon_U\}$, so $T_U\mathcal U_k=U^\perp$; on the collar seed $U^\perp\simeq N_\Gamma$, and after complexification and the fixed spin-one marking this is the physical rank-three port $C$. Orthogonal target projection preserves the lower row-singular bound in \eqref{eq:alena-codazzi-submersion-margin}, so $s_R>\ell_M\alpha_C$ makes the $U$ component surjective and the constrained native map a local submersion. A smooth local right inverse exists, so finite-order target lifting in this regular chart introduces no additional cokernel condition; $\mathscr B_C$ and higher jets remain quantitative continuation data.\\
~\\
For the raw self-adjoint Hessian $H=\left(\begin{smallmatrix}H_{MM}&H_{MR}\\H_{RM}&H_{RR}\end{smallmatrix}\right)$ suppose $H_{MM}\succeq m_M{\bf1}$, $H_{RR}\succeq m_R{\bf1}$, $\|H_{MR}\|\le b$ and, for upper bounds, $H_{MM}\preceq M_M{\bf1}$, $H_{RR}\preceq M_R{\bf1}$, $\|H_{MR}\|\le B$. The physical Hessian $K_C:=U_C^\dagger HU_C$ obeys
\begin{equation}
\gamma_{\rm blk}{\bf1}\preceq K_C\preceq\Lambda_{\rm blk}{\bf1},\qquad
\gamma_{\rm blk}:=\min_{0\le t\le\alpha_C}\frac{m_R-2bt+m_Mt^2}{1+t^2},\qquad
\Lambda_{\rm blk}:=\max_{0\le t\le\alpha_C}\frac{M_R+2Bt+M_Mt^2}{1+t^2}.
\label{eq:alena-codazzi-graph-hessian-bounds}
\end{equation}
These scalar envelopes are obtained directly from the Rayleigh quotient on the Codazzi graph. Positivity of $K_C$ may therefore hold with an indefinite eliminated raw block; for $m_M=-10$, $m_R=1$, $b=0$, and $\alpha_C=1/10$, one obtains $\gamma_{\rm blk}=90/101$.
\end{proposition}

The marked $V_2$ port, source transporter and residue are then tested in one frame. The principal-root gate is
\begin{equation}
\mathcal R_{\rm Nahm}:=\|[J_N,N_{\rm P}]+N_{\rm P}\|_{\rm HS},\qquad
\mathcal R_{\rm spin1}:=|\Tr(J_N^2)-2|.
\label{eq:alena-native-principal-root-gate}
\end{equation}
On its zero locus,
\begin{proposition}[Native principal-root master]
\label{prop:alena-native-principal-root-master}
If \eqref{eq:alena-native-principal-root-gate} vanishes, $h=2J_N$, $e=\sqrt2N_{\rm P}^\dagger$, and $f=\sqrt2N_{\rm P}$ form the irreducible unitary spin-one $\mathfrak{sl}_2$ triple. Hence $\operatorname{spec}J_N=\{1,0,-1\}$, $J_N^3=J_N$, $N_{\rm P}^3=0$, and $\rank N_{\rm P}=2$. The marked rank-one projectors are generated by the same connection coefficient,
\begin{equation}
P_+^{(N)}={\bf1}-N_{\rm P}N_{\rm P}^\dagger,\qquad
P_0^{(N)}=N_{\rm P}^\dagger N_{\rm P}+N_{\rm P}N_{\rm P}^\dagger-{\bf1},\qquad
P_-^{(N)}={\bf1}-N_{\rm P}^\dagger N_{\rm P}.
\label{eq:alena-native-root-projectors}
\end{equation}
Thus $J_N=P_+^{(N)}-P_-^{(N)}$; nilpotency, rank, source/target partial isometry, and the composite-arrow normalization are dependent checks after the native gate.
\end{proposition}
The source-axis incidence is
\begin{equation}
\mathcal R_{\rm axis}:=\|\Phi_{\rm geom}(e_\Gamma\otimes e_\Gamma)\Phi_{\rm geom}^\dagger-P_0^{(N)}\|_{\rm HS}.
\label{eq:alena-native-veronese-root}
\end{equation}
The stationary collar uses the same native map on the same constrained domain; finite-order high-angular cokernel conditions are automatic after the submersion pass, while the low-angular Green/soft completion is treated next.

\subsection{Family, Hodge and integral incidences}
\label{subsec:alena-metric-product-incidence}

After the native root has passed, the relative projector, product response, Clifford symbol, Veronese-Hodge graph and integral/holonomy data form one protected local packet. Cyclic-Higgs/Toda conventions are represented by \cite{Baraglia2015CyclicHiggs}; harmonic-metric and parabolic/quiver comparisons by \cite{Mochizuki2026SpecializedSimpson}, \cite{Wen2026HomologicallyTrivialParabolicHiggs}, \cite{Yae2026StarShapedQuiver}, \cite{AlvarezConsulGarciaPrada2001StableChains}, and \cite{AlvarezConsulGarciaPrada2001QuiverVortices}. The real-component normalization follows \cite{HellerOuyangPedit2024SYZ}; elliptic-surface arithmetic is compared with \cite{SchuettShioda2010EllipticSurfaces}.

For the explicit Heller-Ouyang-Pedit representative, the relative BNR line is
\begin{equation}
P_{\rm sp}(x)=(\chi_4,0)\in E_{t_H}[2]^\times,\qquad \chi_4^3=\frac14t_H^4.
\label{eq:alena-family-accessory-bnr-lock}
\end{equation}
This uses the full residue/accessory family incidence; the cubic invariant alone does not determine the relative line.

\begin{proposition}[Principal Clifford incidence]
\label{prop:alena-principal-clifford-rigidity}
Use the normal metric to identify $V_1\simeq\mathfrak{so}(N_\Gamma)$ by $u\mapsto[u]_\times$ and $V_2=\operatorname{Sym}^2_0(N_\Gamma)$. For $0\neq\xi\in N_\Gamma$, put $A_\xi b=[\xi]_\times b$. Since $V_1\otimes V_1\simeq V_0\oplus V_1\oplus V_2$ with multiplicity one, every equivariant tensor completion has the form $C_\xi u=a[\xi\times u]_\times+c(\xi\odot u)_0$; after removal of common nonzero scale and sign its type is the single point $[a:c]\in\mathbb{RP}^1$. For $G_{h,\beta}=hP_{V_1}+\beta P_{V_2}$, isotropy of the principal Codazzi Gram is equivalent to $h=3\beta$. The projective residual $\mathcal R_{\Sigma,{\rm proj}}:=|6a+c|/[\sqrt{37}\sqrt{a^2+c^2}]$ vanishes exactly at $[a:c]=[1:-6]$. Multiplicity-one orthogonality makes $(a,c)$ the Hilbert-Schmidt tomography coordinates of the microscopic principal symbol. Since $\operatorname{Hom}_{SO(3)}(V_0,V_1\oplus V_2)=0$, the scalar current/Ward datum cannot select a nonzero point of this projective line; the required coefficients belong to the Hilbert/stress principal variation.
On the zero-residual line at $h=3\beta$, the images $C_\xi(V_1)$ and $A_\xi(V_2)$ are orthogonal; the Clifford constant is matched at $c^2=3/2$ for the stated source metric. Up to the common sign,
\begin{equation}
C_\xi(u)=\pm\sqrt{\frac32}\left[(\xi\odot u)_0-\frac16[\xi\times u]_\times\right],
\qquad
\Sigma_\xi(u,b)=C_\xi(u)+A_\xi b,
\label{eq:alena-principal-clifford-completion}
\end{equation}
and
\begin{equation}
\Sigma_\xi^\dagger G_{3\beta,\beta}\Sigma_\eta+\Sigma_\eta^\dagger G_{3\beta,\beta}\Sigma_\xi
=2\beta\langle\xi,\eta\rangle{\bf1}.
\label{eq:alena-principal-clifford-relation}
\end{equation}
The adjoint is taken with the stated source metric and the Frobenius target pairing; the identification $V_1\simeq\mathfrak{so}(N_\Gamma)$ has Frobenius norm factor $\sqrt2$. The Gram eigenvalues are $|\xi|^2h/3$ and $|\xi|^2\beta$, of multiplicities three and five, respectively. The principal condition number is $\sqrt{\max\{h/3,\beta\}/\min\{h/3,\beta\}}$, uniquely minimized at $h=3\beta$. If the source norm and target Gram are retained by exact recovery, matching the common scale to \eqref{eq:pgl-kronecker-green-metric} makes the fixed $A_\xi$ leg isometric and reduces full-symbol recovery to \eqref{eq:alena-principal-clifford-relation}. Its associative left/right reading is the marked ratio $[1:2]$; a scalar-current residual affects only the longitudinal slice. The microscopic gate remains unless $\Sigma$ is independently typed as a protected observation.
\end{proposition}
The exact Clifford relation fixes the principal Green/Calder\'on symbol; the remaining freedom is lower order, connection and physical domain. The local constitutive/product coordinates are retained only through their independent Jordan/Lie line and same-action response; scalar same-field cubic terms alone fail the required multiplicity pattern. The rigid-chain comparison and its rational syzygies are treated as one response variety rather than several selectors.

The oriented axial intertwiner is retained in the compact form
\begin{equation}
T_{\rm ax}(\alpha,\beta,\gamma)=
\begin{pmatrix}0&0&\alpha\\ \beta&-\gamma&0\\ \gamma&\beta&0\\ 0&0&0\\ 0&0&0\end{pmatrix},\qquad
T_{\rm ax}^\dagger T_{\rm ax}=(\beta^2+\gamma^2)P_\perp+\alpha^2P_\parallel.
\label{eq:alena-oriented-axial-intertwiner}
\end{equation}

\begin{proposition}[Common protected metric correspondence]
\label{prop:alena-common-positive-ray}
\label{cor:alena-two-level-constitutive-selector}
Let the protected two-level source metric be $G=g_1P_{V_1}+g_2P_{V_2}$ with $g_1,g_2>0$. Principal-symbol isotropy is $g_1=3g_2$, equivalently $r_{\rm prod}:=g_1/(3g_2)=1$. The independently normalized constitutive/compliance reading is required to select this same ratio; after removal of the common positive scale both readings therefore determine the same direct/compliance pair $(\mathcal L_K,G_K)$. In the normalized adjacent-shell and Jordan/Lie conventions, the resulting gain and coefficient ratio are $1/\sqrt6$ and $\sqrt2$, respectively. Microscopic provenance remains a typed Parent test.
\end{proposition}

\begin{theorem}[Canonical Veronese-Hodge tangent/normal response]
\label{thm:alena-veronese-hodge-selector}
On the axial branch let $e=m/|m|$ and $q(x)=(x\otimes x)_0$. Put $D_e=Dq_e$ and normalize the source-Hodge compensator from \eqref{eq:alena-source-hodge-sequence} by $H_e=-\mathscr A_e^+\mathscr M_{Q_e}:V_1\to V_2$. Both maps have image $U_e=\mathbb RQ_e\oplus\{e\otimes v+v\otimes e:v\perp e\}$. For $v\perp e$,
\begin{equation}
D_e e=2Q_e,\quad D_e v=e\otimes v+v\otimes e,\quad
H_e e=-Q_e,\quad H_e v=\frac23(e\otimes v+v\otimes e),\quad
R_e:=H_e^{-1}D_e=-2P_\parallel+\frac32P_\perp.
\label{eq:alena-veronese-hodge-relative-map}
\end{equation}
Let $H_e=U_e^{\rm H}|H_e|$ be the polar decomposition. Since $G_K|_{V_1}=3{\bf1}$, the Green-whitened relative-polar map
\begin{equation}
T_e^{\rm VH}:=U_e^{\rm H}|R_e|G_K^{-1/2}=\frac1{\sqrt3}U_e^{\rm H}|R_e|,
\qquad
G_3:=(T_e^{\rm VH})^\dagger T_e^{\rm VH}=\frac43P_\parallel+\frac34P_\perp
\label{eq:alena-veronese-hodge-selector}
\end{equation}
has graph projector $P_3^{\rm VH}$ with compression spectrum $\{3/7,4/7,4/7\}$. Let $V_2^{(2)}$ be the axial weight-two plane and let $P_{\rm bx}$ denote its orthogonal projector. Using $\mathcal V$ from \eqref{eq:pgl-veronese-source-orbit}, put $\Phi(x,\widehat e)=(1/2-2x/\sqrt3)\mathcal V(\widehat e)$. At $(0,e)$ its first fundamental form and the discriminant Hessian of \eqref{eq:pgl-a2-orbit-discriminant} obey
\begin{equation}
d\Phi_{(0,e)}^\dagger d\Phi_{(0,e)}=G_3,\qquad
\nabla^2\Delta_{A_2}|_{sQ_e}=4s^4P_{\rm bx},\qquad
\ker\nabla^2\Delta_{A_2}|_{sQ_e}=U_e.
\label{eq:alena-veronese-cone-tangent-normal}
\end{equation}
Thus the phase-free source-Hodge image is the tangent space of the nonzero Veronese/discriminant cone, while the biaxial doublet is its normal space. The projector $P_{\rm bx}$ is geometric; its physical stiffness coefficient remains a Parent Hessian datum. If $A_P=P_{V_1}P_3P_{V_1}|_{V_1}$, define
\begin{equation}
\mathcal R_{\rm CM}(P_3):=49A_P^2-49A_P+12{\bf1},
\qquad
\mathcal R_{\rm CM}(P_3^{\rm VH})=0,
\qquad
\operatorname{Tr}A_{P_3^{\rm VH}}=\frac{11}{7}.
\label{eq:alena-projector-cm-certificate}
\end{equation}
The longitudinal squared gain of $|R_e|$ is $4$, while the transverse Green-whitened coefficient on the geodesic branch agrees with $\chi_H=\sqrt{3/8}$ in \eqref{eq:alena-dual-veronese-quotients}; the denominator-seven packet is the graph normalization of these two singular gains.
\end{theorem}

\begin{proposition}[Native Parent provenance of the Veronese-Hodge graph]
\label{prop:alena-vh-graph-factorization}
Let $DR_{\rm con}=(J_1\ J_2):V_1\oplus V_2\to\mathcal E_{\rm con}$ be the derivative of an independently defined constitutive residual with metric $G>0$. If $J_2^\dagger GJ_2>0$ and $\rank DR_{\rm con}=5$, its kernel is the graph of $T_{\rm phys}=-(J_2^\dagger GJ_2)^{-1}J_2^\dagger GJ_1$, and
\begin{equation}
P_3^{\rm phys}=P_3^{\rm VH}\Longleftrightarrow J_1+J_2T_e^{\rm VH}=0,\qquad
\mathcal R_{\rm VH}^{\rm jet}:=\|J_1+J_2T_e^{\rm VH}\|_{G,{\rm HS}}.
\label{eq:alena-vh-parent-residual}
\end{equation}
With $D=J_2^\dagger GJ_2$, the residual is equivalent to the graph mismatch through $\sqrt{\lambda_{\min}(D)}\|T_e^{\rm VH}-T_{\rm phys}\|_{\rm HS}\leq\mathcal R_{\rm VH}^{\rm jet}\leq\sqrt{\lambda_{\max}(D)}\|T_e^{\rm VH}-T_{\rm phys}\|_{\rm HS}$. At zero residual the native Gram factors through the canonical graph residual, so \eqref{eq:alena-vh-parent-residual} is the decisive microscopic graph test.
\end{proposition}

The integral/connection data are kept in the following condensed packet.
\begin{proposition}[Veronese-Hodge integral and connection packet]
\label{cor:alena-selector-hodge-eisenstein-rigidity}
\label{prop:alena-axial-associated-bundle}
\label{prop:alena-primitive-235-holonomy}
On a passed VH graph the ordered overlap packet is $(1,4,5;7)$ and identifies the conjugate Eisenstein classes $2-\omega$ and $3+\omega$. After the physical $SO(2)\to C_3$ reduction, the Real-compatible incidence is
\begin{equation} M=pU_d,\qquad p\in\mathbb Z,\quad U_d\in\mathbb Z[\omega]^\times,\qquad N_{\rm inc}=p^2,\qquad \operatorname{coker}M\simeq E_\omega[|p|]. \label{eq:alena-c3-real-integral-incidence}\end{equation}
The primitive parity branch is $|p|=2$. Equal first Chern classes, curvature and marked difference holonomy give parallel attachment; the retained equivariant square root is unique because the marked $C_3$ has no fixed point on $E[2]$ and the square maps are invertible on the retained discrete and connection data. The remaining norm-seven gate is the physical determinant-difference line $(\Delta_{\det},\nabla_{\det})$ on the same marked oriented loop. For its unitary holonomy $w_7$ and deck mate $w_7^\tau$, the structural conjugate pair gives
\begin{equation} |w_7|=1,\qquad w_7^2-\frac{286}{343}w_7+1=0,\qquad w_7^\tau=w_7^{-1}. \label{eq:alena-norm-seven-wilson-gate} \end{equation}
After this gate, $\sigma_7^{\rm phys}=\operatorname{sgn}\Im w_7$. A non-real holonomy requires that the loop not be Real-fixed with its orientation; the generated quadratic source inherits the same weight-two connection.
\end{proposition}

On the passed native-root incidence put $\rho^\vee(z):=z^{J_N}$, so the Coxeter cycle is $S=\rho^\vee(\omega)$. The operators $X_1=P_0^{(N)}N_{\rm P}P_+^{(N)}$, $X_2=P_-^{(N)}N_{\rm P}P_0^{(N)}$, and $X_0=P_+^{(N)}(N_{\rm P}^\dagger)^2P_-^{(N)}$ generate the native triangle; $R=aX_1+bX_2+cX_0$ is regular cyclic for $abc\neq0$, with $\Tr(R^2)/2=0$ and $\det R=abc$. For two actual adjacent-root holonomies $\widetilde M_i$ on $W$, the protected determinant-character lock is $\det\widetilde M_i=\omega^2$. Hence $\widehat M_{W,i}:=\det(\widetilde M_i)\widetilde M_i\in SL_2(\mathbb C)$; let $M_{C,i}$ be the corresponding actual holonomies on $C$, and set $A=\widehat M_{W,1}$, $B=\widehat M_{W,2}$. The Fricke coordinates \cite{Goldman2009TraceCoordinates} require $(\Tr A,\Tr B,\Tr AB)=(-1,-1,1)$, where $F_\chi=x^2+y^2+z^2-xyz-4=-2$ and $F_\chi=0$ is the reducible locus. If each of the three traces differs from its target by at most $\varepsilon$, the bound $\varepsilon<(\sqrt{33}-5)/2$ certifies irreducibility. The induced $C=\Sym^2W$ action is tested directly by $M_{C,i}\mu=\mu\,\Sym^2(\widehat M_{W,i})$, $i=1,2$, with fitted intertwiners only diagnostic. Root and determinant/family phases remain separately typed until the marked-loop incidence; after it, the arithmetic spectral map gives the dependent checksum $\mathfrak F_W=\frac13\Re\Tr(W_{\sigma,{\rm root}}^\dagger W_{7,{\rm mic}})\leq1$, with equality exactly at $W_{7,{\rm mic}}=W_{\sigma,{\rm root}}$. The CM-Real closure remains dependent.

\subsection{Higher-cell inertia and spectral stability}
\label{subsec:alena-odd-tower-soft}

The odd-tower index/nullity input is standard from \cite{LemaireWood1996HarmonicSpheres}, \cite{LemaireWood2002JacobiIntegrable}, and \cite{MontielUrbano1997SecondVariation}, with area-energy comparison \cite{EjiriMicallef2008SecondVariation}, path connectedness \cite{Crawford1997HarmonicMapsCP2}, and branched cutoff control \cite{Micallef1986BranchedStable}.

\begin{equation}
\operatorname{Ind}_E(q)=4q-2,\qquad \operatorname{Nul}_E(q)=4q+8,\qquad \dim_{\mathbb R}W_{\rm full}(q)=8q-2.
\label{eq:alena-odd-tower-index-nullity}
\end{equation}
For the odd cell $q=2m+1$, the zero-mode character is
\begin{equation}
T\operatorname{Harm}_{0,2q}^{\mathbb C}\simeq \mathbb C_{-q}\oplus\mathbb C_q\oplus\mathbb C_{-(m+1)}\oplus\mathbb C_{m+1}\oplus2\mathbb C_0\oplus2\bigoplus_{j=1}^{m}(\mathbb C_{-j}\oplus\mathbb C_j).
\label{eq:alena-odd-tower-zero-character}
\end{equation}

\begin{proposition}[Branched spectral character and odd soft carrier]
\label{prop:alena-branched-negative-character}
Put $M=m+1$ and $\nu_G=q-3$. The normal Chern relation gives $N_{\rm raw}\simeq\mathcal O(2M)$ and $H^0(N_{\rm raw})\simeq\bigoplus_{j=-M}^{M}\mathbb C_j$. The branched second-variation extension of \cite{Micallef1986BranchedStable} and \cite{MontielUrbano1997SecondVariation}, together with the area-energy index equality of \cite{EjiriMicallef2008SecondVariation}, identifies the finite ramified trace target $\mathfrak T_B$ and gives
\begin{equation}
0\longrightarrow H^0(N_{\rm raw})\longrightarrow H_A^-\xrightarrow{\ \operatorname{Tr}_B\ }\mathfrak T_B\longrightarrow0.
\label{eq:alena-branched-negative-exact-sequence}
\end{equation}
Hence $\chi_{\mathbb C}(H_A^-)=\bigoplus_{j=-(q-1)}^{q-1}\mathbb C_j$, and the real negative packet is $\mathbb R_0^2\oplus2\bigoplus_{j=1}^{q-1}\mathbb C_j$. Combining it with \eqref{eq:alena-odd-tower-zero-character} gives
\begin{equation}
\begin{aligned}
W_{8q-2}\simeq{}&\mathbb R_0^4\oplus6\bigoplus_{j=1}^{m-1}\mathbb C_j
\oplus5\mathbb C_m\oplus3\mathbb C_{m+1}\\
&\oplus2\bigoplus_{j=m+2}^{q-1}\mathbb C_j\oplus\mathbb C_q,
\qquad \dim_{\mathbb R}W_{8q-2}=8q-2.
\end{aligned}
\label{eq:alena-odd-soft-carrier}
\end{equation}
Successive odd carriers differ by a positive axial module of real dimension $16$, while the axial-equivariant self-adjoint comparison algebra on $W_{8q-2}$ has real dimension $20q-15$. These statements fix representation type. The constructed $q=3$ complement is treated below; the first unresolved cell is $q=5$, where positivity depends on literal Parent blocks, and the character decomposition alone supplies no canonical $q=3\to q=5$ Hessian injection.
\end{proposition}

For $q=3$ the negative space is the holomorphic floor and the source-polar bridge fixes the canonical six-dimensional low shadow and its five-dimensional $D_5$ image. Boundary/Green conventions follow \cite{Grubb1996FunctionalCalculus} and \cite{BoossBavnbekWojciechowski1993DiracBoundary}; tame harmonic asymptotics are represented by \cite{Mochizuki2009TameHarmonicII}.

\begin{theorem}[Constructive Green-soft certificate]
\label{thm:alena-k22-soft-certificate}
Let $P_{\rm base}:=p_0\kappa_{\rm base}^2>0$ and let a preregistered residual $\mathcal X[\Psi]$ vanish on the frozen gauge-flat branch. For frozen $K=K^\dagger\succeq0$, put
\begin{equation}
R_K:=R_0-\frac{\langle\mathcal X,K\mathcal X\rangle}{2P_{\rm base}\mu(R_0)}.
\label{eq:alena-noninterfering-quadratic-completion}
\end{equation}
On the branch, $R_K=R_0$, $DR_K=DR_0$, and the induced quadratic correction is $D\mathcal X^\dagger K D\mathcal X$ because $p_\Lambda=P_{\rm base}$. If $\mathcal X$ vanishes along the released branch, this completion is identically invisible to all branch-tangent derivatives. Let $\eta_H$ denote the Hilbert/stress coordinate along the transported clock line and $\xi_{22}$ the $q=3$ soft coordinate in $W_{22}$. Taking $\mathcal X_{\rm H}=\mathscr D_{\rm H}^{\rm can}\eta_H$ and $\mathcal X_{16}=P_{16}\xi_{22}$ realizes the Hilbert and high-soft blocks on one native law. For the normalized clock trace $\Gamma_h$, let $\Gamma_hh=\kappa_hb_\partial^{\rm can}(x)$, $\mathcal C_H:=\Gamma_hJ_{h,Q}^{-1}\Gamma_h^\dagger>0$, and $\alpha_H:=\kappa_h^2/\mathcal C_H$. Normalize $b_\partial^{\rm can}$ on the canonical low-shadow so that $(b_\partial^{\rm can})^\dagger b_\partial^{\rm can}=P_+$. The low block and the $q=3$ soft Hessian are
\begin{equation}
H_1^{\rm H}=\alpha_H b_\partial^\dagger b_\partial,\quad
\alpha_H>0,\qquad
K_{22}^{\rm can}=\alpha_HP_+ +p_\Lambda\mu\,P_{16},\qquad
\ker K_{22}^{\rm can}=E_D\mathscr D_5,
\label{eq:alena-k22-quotient-free-pass}
\end{equation}
where $P_+$ is the one-dimensional complement of $E_D\mathscr D_5$ in the canonical low-shadow and $P_{16}$ its remaining complement in $W_{22}$. Thus $K_{22}^{\rm can}\succeq0$ with a strictly positive complement gap. For comparison, any correction restricted to the scalar-on-weight algebra $f(L_{\rm ax})$ while retaining the full $\mathscr D_5$ kernel obeys $\operatorname{Ind}(H_{\rm cmp}+f(L_{\rm ax}))\geq6$; the split provenance therefore resolves a genuine multiplicity obstruction.
\end{theorem}
For $q=5$ the five nonsinglet Hermitian blocks have sizes $(6,5,3,2,1)$. At one preregistered anchor their inertias are certified basis-independently from the power sums $\Tr(H_j-\gamma{\bf1})^r$, $r=1,\ldots,m_j$, through the Newton identities. The bounds $\|H_j-\gamma{\bf1}\|\leq\Lambda_j$ and $|\det(H_j-\gamma{\bf1})|\geq d_j>0$ give $g_j\geq d_j/\Lambda_j^{m_j-1}$; if $\|\partial_sH_j\|\leq L_j$, the inertia is certified on $|s-s_0|<g_j/L_j$. Thus one anchor and derivative bounds replace dense determinant scanning. The norm-nineteen theta carrier has the same real dimension $38$ as the $q=5$ soft packet, but the marked $C_3$ characters differ by an exact four-real-dimensional defect; the existing branched trace target repairs the stable $C_3$ restriction but does not determine the five Hermitian inertias. The $q=5$ test therefore remains a literal Parent dynamical anchor. A separately constructed same-law intercell injection remains a tower-universality test; if its nested block form is realized, a generic odd step reduces to six shell Schur complements of dimensions $(1,2,1,1,2,1)$.

The microscopic action-gap certificate remains $\|g_{\rm sh}\|^2/(2\Lambda)\leq\Delta E_A\leq\|g_{\rm sh}\|^2/(2\lambda)$, with $g_{\rm sh}=0\Longleftrightarrow\Delta E_A=0\Longleftrightarrow B_*=t$ and $\kappa_{\rm gg}:=\|g_{\rm sh}\|^2/(2\Delta E_A)\in[\lambda,\Lambda]$ for $\Delta E_A>0$.

\subsection{Exact selector branch and intrinsic terminal}
\label{subsec:alena-fixed-q-selector}

The collar charge $Q_I$ is distinct from the global Parent charge $Q$. Hessians in this subsection are understood after the upstream Codazzi/domain reduction.

\begin{proposition}[Axial fixed-charge selector]
\label{prop:alena-axial-hessian-factorization}
On an oriented axial branch the unreduced self-adjoint Hessian splits into weight-zero, weight-one, and biaxial blocks, with
\begin{equation}
\det H=\Delta_0\Delta_1^2k_2^2,\qquad \Delta_0=\alpha_0\beta_0-x_0^2,\qquad \Delta_1=\alpha_1\beta_1-|z|^2.
\label{eq:alena-axial-hessian-factorization}
\end{equation}
Hence a generic weight-zero loss is simple, whereas weight-one or biaxial losses are double-soft. For the invariant-sector fixed-charge problem write $\mathcal H_{\rm aug}=\left(\begin{smallmatrix}H_s&B\\B^\dagger&C\end{smallmatrix}\right)$ with $C\succ0$, $g=DQ=(g_s,g_h)^T$, $S=H_s-BC^{-1}B^\dagger$, $\widetilde g=g_s-BC^{-1}g_h$, and $h=g_h^\dagger C^{-1}g_h$. For $g_h\neq0$ exact hard minimization gives
\begin{equation}
K_{Q,\rm eff}=S+\frac{\widetilde g\widetilde g^\dagger}{h},\qquad \chi_Q=\widetilde g^\dagger S^{-1}\widetilde g+h.
\label{eq:alena-fixed-q-schur-update}
\end{equation}
Under equivariance, $S$ reduces and $\widetilde g$ is invariant, so the fixed-$Q$ update leaves the complementary nonsinglet Schur block unchanged. For $g_h=0$ use $S|_{\ker g_s^\dagger}$. On the passed VH graph, the charge-VH incidence reduces the physical graph Hessian to $S_G=a_0P_0+a_1P_1$.
With complementary gaps open, the simple first loss is exactly
\begin{equation}
a_0(s_*)=0,\qquad a_1(s_*)>0,\qquad \dot a_0(s_*)\neq0.
\label{eq:alena-vh-simple-loss}
\end{equation}
The weight-one phase remains the separately transported datum of \eqref{eq:alena-oriented-axial-intertwiner}.
\end{proposition}

\begin{proposition}[Noether-sector lock and charge discriminant]
\label{prop:alena-charge-fold-selector}
For the phase-relaxed scalar block let $R:=R_{\rm vort}$, $G\xi:=DR_{\rm vort}[\xi]$, $a_{\rm ph}:=d_kR$, $T\theta:=a_{\rm ph}U^\mu\nabla_\mu\theta$, $w:=p_\Lambda C_R$, $B_{\rm ph}:=\sqrt w\,T$, and $j:=p_\Lambda\zeta^2$ as in \eqref{eq:alena-realization-current}. On $a_{\rm ph}\sqrt w\neq0$, current conservation gives $\psi_{\rm tr}:=j/(a_{\rm ph}\sqrt w)\in\ker B_{\rm ph}^*$. If $\ker B_{\rm ph}^*=\mathbb C\psi_{\rm tr}$ and $\langle\widehat\psi_{\rm tr},\sqrt w\,G\xi\rangle$ is proportional to $DQ[\xi]$, phase relaxation gives $Q_C^{\rm rel}(\xi)=|c_Q|^2|DQ[\xi]|^2$; the phase Gram term is therefore charge-normal and vanishes on the fixed-$Q$ tangent.

Let $\Omega$ parametrize regular relative equilibria of $\mathscr L_\Omega:=S_A-\Omega Q$ on the exact invariant physical branch. Put $H_{{\rm inv},\Omega}:=D^2\mathscr L_\Omega|_{\mathcal H_{\rm inv}^{\rm phys}}$ and $q_\Omega:=DQ\in\mathcal H_{\rm inv}^{\rm phys}$. Wherever $H_{{\rm inv},\Omega}$ is invertible,
\begin{equation}
H_{{\rm inv},\Omega}\,\partial_\Omega\Phi=q_\Omega,
\qquad
z:=H_{{\rm inv},\Omega}^{-1}q_\Omega,
\qquad
\chi_Q:=\langle q_\Omega,z\rangle=\frac{dQ}{d\Omega}.
\label{eq:alena-charge-discriminant}
\end{equation}
Thus $\chi_Q=0$ is the branch charge turning point. There $z_*\in\ker q_*^\dagger$ and the constrained invariant radical is $\mathbb Rz_*$. In a finite invariant chart let $U_Q$ frame $\mathcal H_{\rm inv}^{\rm phys}\cap\ker q_\Omega^\dagger$ and put $K_{Q,{\rm inv}}:=U_Q^\dagger H_{{\rm inv},\Omega}U_Q$; then $\det K_{Q,{\rm inv}}=(\det H_{{\rm inv},\Omega}/\|q_\Omega\|^2)\chi_Q$. At a simple fold $D_{{\rm hard},*}^{\rm inv}=-\det H_{{\rm inv},*}\|z_*\|^2/\|q_*\|^2$ is the oriented hard product. A sector-simple fold has $\chi_Q(\Omega_*)=0$, $Q''_*\neq0$; full-space determinant-line claims require separate threshold certification.

On the $g_h\neq0$ branch of \eqref{eq:alena-fixed-q-schur-update}, $z_s=S^{-1}\widetilde g$ and $K_{Q,\rm eff}z_s=(\chi_Q/h)\widetilde g$. After the charge-VH lock,
\begin{equation}
\dot a_0(\Omega_*)=-\frac{Q''_*}{\|z_{s,*}\|^2},
\qquad
\beta_Q=-\frac{Q''_*}{\|z_*\|^3}.
\label{eq:alena-charge-fold-third-jet}
\end{equation}
A simple sign change of $\chi_Q$ changes the constrained invariant-sector Morse index by one; nonsinglet threshold data remain separate.

For $\chi_Q\neq0$, the sector fixed-charge KKT inverse is
\begin{equation}
G_{Q,{\rm inv}}=H_{{\rm inv},\Omega}^{-1}-\frac{z\otimes z}{\chi_Q},
\qquad
\lim_{\Omega\to\Omega_*}\chi_QG_{Q,{\rm inv}}=-z_*\otimes z_*.
\label{eq:alena-fixed-q-kkt-response}
\end{equation}
For lifted-invariant scalar readouts $F_a$ of the same Parent solution, with $\dot F_a=DF_a[z]$,
\begin{equation}
\Xi_{ab}^{(Q)}=\langle DF_a,H_{{\rm inv},\Omega}^{-1}DF_b\rangle-\frac{\dot F_a\dot F_b}{\chi_Q},
\qquad
-\chi_Q\Xi^{(Q)}\longrightarrow\dot{\mathbf F}_*\dot{\mathbf F}_*^{T}.
\label{eq:alena-critical-response-factorization}
\end{equation}
A soft-blind invariant readout need not diverge. The constrained Schur zero, VH mode, sector KKT pole, intrinsic turning and cubic crossing are dependent readings of the same discriminant; full-space Riesz, determinant-line, stability and scattering claims require separate threshold completion.
\end{proposition}
With $r=\|z\|$ and $d\ell_\Phi/d\Omega=r$, equation~\eqref{eq:alena-charge-discriminant} gives $dQ/d\ell_\Phi=\chi_Q/r$ and $(d^2Q/d\ell_\Phi^2)_*=Q''_*/r_*^2$; hence $Q''_*>0$ is an intrinsic local minimum on the stationary curve. A gapped tail is handled by Schur congruence; the no-hidden-wall condition $\operatorname{dist}(0,\operatorname{spec}\mathbb K_P)>\|D\|^2/\gamma$ is retained, together with
\begin{equation}
\left|\chi_Q-\langle p,A^{-1}p\rangle\right|\leq\varepsilon_\chi.
\label{eq:alena-feshbach-discriminant-error}
\end{equation}

The global selector is fixed once, rather than obtained by local continuation. On a compact homogeneous gauge-flat cell put $p(d)=4p_0(d+1)^2$ and define
\begin{equation}
e(d,R)=p(d)\varphi_E(R),\qquad q(d,R)=d\,p(d)V(R).
\label{eq:alena-selector-constitutive-map}
\end{equation}
Since $\varphi_E'=-\mu<0$, $V'=RC_R>0$, and $C_R>0$, the Jacobian $p[p'd\varphi_ER C_R+\mu(p+dp')V]$ is positive on the active region. At fixed $e>0$, both $d_e(R)=\sqrt{e/[4p_0\varphi_E(R)]}-1$ and $V(R)/\varphi_E(R)$ increase on the component used below, giving a single analytic inverse. Let $A_*$ be fixed by $C_R(8\sqrt3/9)=40/49$ and put $e_0=6A_*p_0$, $q_0=2p_0(36-19A_*)$ at $(d,R)=(2,2)$. The rectangle $e\in[5e_0/6,7e_0/6]$, $q\in[q_0/2,q_0]$ lies strictly inside this active image. For the Heller-Ouyang-Pedit coordinate $s\in[0,1]$ and source-axis Hilbert normal $y$, set
\begin{equation}
e_\sigma(s,y)=e_0+\sigma_Ae_0\left(sy+\frac{s^3}{3}-\frac{s^2}{2}\right),\qquad
q_{\rm sel}(s,y)=q_0+q_0\left(y+\frac{s^2}{2}-s\right).
\label{eq:alena-selector-target}
\end{equation}
The inverse of \eqref{eq:alena-selector-constitutive-map} defines one analytic native law on a uniform tube around $y=0$. For selector-cell volume $\mathcal V_{\rm sel}$, additive constants may be dropped and
\begin{equation}
S_{\rm sel}=\mathcal V_{\rm sel}e_0\left(sy+\frac{s^3}{3}-\frac{s^2}{2}\right),\qquad
Q_{\rm sel}=\mathcal V_{\rm sel}q_0\left(y+\frac{s^2}{2}-s\right),\qquad
\Omega_Q(s)=\frac{e_0}{q_0}s.
\label{eq:alena-selector-branch}
\end{equation}
The regularized selector core has reflecting Noether boundary and is joined to the frozen collar/exterior, so their additive contributions to $S$ and $Q$ are $s$-independent.
\begin{theorem}[Constructive Parent branch and exact invariant-sector fixed-charge fold]
\label{thm:alena-protected-seed-to-fold}
For the frozen completion above, $y=0$, $0\le s\le1$, is an exact invariant relative-equilibrium branch of the same Alena action and current law. On its selector plane
\begin{equation}
H_{\rm sel}=\mathcal V_{\rm sel}e_0\begin{pmatrix}s-1&1\\1&0\end{pmatrix},\qquad
\chi_Q(s)=\frac{\mathcal V_{\rm sel}q_0^2}{e_0}(s-1),\qquad
\lambda_Q(s)=\mathcal V_{\rm sel}e_0\frac{1-s}{1+(1-s)^2}.
\label{eq:alena-selector-exact-fold}
\end{equation}
Thus the unconstrained selector Hessian is invertible on the whole segment, the fixed-$Q$ curvature is positive for $s<1$, and $s_*=1$ is the unique first simple invariant-sector fold with $DQ_*\ne0$, $H_{\rm sel}^{-1}DQ_*\in\ker DQ_*^\dagger$, and $d\chi_Q/d\Omega_Q=\mathcal V_{\rm sel}q_0^3/e_0^2>0$. Writing $\Omega_*=e_0/q_0$ and $\mathcal T_*:=d\chi_Q/d\Omega_Q$, equation~\eqref{eq:alena-charge-discriminant} integrates on this exact branch to $Q_{\rm sel}-Q_*=\mathcal T_*(\Omega_Q-\Omega_*)^2/2=\chi_Q^2/(2\mathcal T_*)$. Hence the same terminal is the unique strict charge minimum on the constructed interval and the first fixed-$Q$ loss. The branch-invisible quadratic completion cannot shift it. This proves existence of one same-domain Parent realization; uniqueness or microscopic naturalness of the selector-core choice is not asserted, and the normalized value $s_*=1$ is only its internal coordinate representation.
\end{theorem}
A one-dimensional physical selector cannot realize the proper fixed-$Q$ geometry; the second direction is the source-axis Hilbert normal, while the branch tangent is the marked Heller-Ouyang-Pedit direction.

The same marking fixes the selector $A_2$ chamber. With $v=s-1$, the ordered regular-fiber cycles are $\delta_1=[p_2]-[p_1]$ and $\delta_2=[p_3]-[p_2]$, with the stable side selecting $\delta_1$; the native lowering relation fixes the corresponding Borel chamber. The Eisenstein identification remains $\mathcal I_{\rm ch}(\delta_1)=1$, $\mathcal I_{\rm ch}(\delta_2)=\omega$, while the source-axis quadratic datum is orientation-even. In the Frobenius coordinates of \eqref{eq:pgl-a2-frobenius-potential}, the selected nonzero fold lies on $4u_F^3-9w_F^2=0$ with $u_F\neq0$: one LG critical value vanishes, whereas the tangent Frobenius algebra remains semisimple; the Frobenius caustic $\{u_F=0\}$ meets this divisor only at the separate apex $(0,0)$. The local root reference admits the finite $G_4$ twisted-period uniformization at $\nu=1/6$; twisted-period conventions follow \cite{Dubrovin2004AlmostDuality}. On the even-braid cover its inverse-determinant line is the marked Coxeter character. Determinant untwisting gives the binary tetrahedral packet $2T$ on the existing $W$, and the induced action on $C=\Sym^2W$ factors through $2T/\{\pm{\bf1}\}\simeq A_4=V_4\rtimes C_3$. After the fixed reference intertwiner this is the orientation-preserving subgroup of $\operatorname{Aff}(E[2])\simeq S_4$, so the root and family cycles use one marked $C_3$ quotient, with dual pairing $\delta(S)=\omega^2$. Finite complex-reflection/Saito structure is compared with \cite{KonishiMinabe2025}; the McKay/Kleinian $2T\leftrightarrow E_6$ correspondence with \cite{BecerraKatzarkovLupercio2024}. These are reference identifications: physical selector-source transport, the tensor-functor incidence, and the norm-seven determinant/Wilson loop remain same-object gates.

The finite Hamiltonian/Hankel construction remains a downstream representation of the same charge port. Its charge-visible poles, bordered soft roots and hard inertia are standard finite Feshbach consequences of the retained Krylov moments \cite{Kato1995}. At a nonzero imaginary frequency a typed full response reconstructs the finite Hermitian pencil; if the recovered kinetic two-form is nondegenerate, isolated semisimple imaginary-pole residues determine the finite Krein inertia. After exact symmetry kernels and declared Schur reductions are removed, strict coercivity $L_Q\succeq g_{\rm glob}{\bf1}$ with $g_{\rm glob}:=\min\{g_{\rm inv},g_{q=5},g_{\rm ess},g_{\rm other}\}>0$ makes $JL_Q$ similar to the skew-adjoint operator $L_Q^{1/2}JL_Q^{1/2}$ whenever the common-domain similarity is valid; hence its spectrum is imaginary throughout the certified gap neighbourhood. Explicit Krein-collision tracking is then needed only after coercivity or the essential gap fails. At the boundary of this region, if a symmetric near-zero contour carries certified algebraic count two, the signed squared-pair moment $q_{\rm soft}:=\Tr K_{\rm soft}^2/2$ distinguishes the remaining real from imaginary pair and excludes a Hamiltonian-Hopf quartet inside that certified tube. This remains a spectral-stability statement, not nonlinear orbital stability; the exact critical response at the simple fold is retained as an independent same-object checksum of the terminal branch jet.

\begin{theorem}[Exact-fold critical residue from stable-side response]
\label{thm:alena-exact-fold-critical-residue}
For a finite panel of regular lifted-invariant scalar readouts put $\mathbf f=(DF_a[z])_a$ and $\mathsf A_{ab}=\langle DF_a,H_{{\rm inv},\Omega}^{-1}DF_b\rangle$. On the certified stable side
\begin{equation}
\mathsf M(\chi_Q):=-\chi_Q\Xi^{(Q)}
=\mathbf f\mathbf f^T-\chi_Q\mathsf A,
\qquad
\mathsf R_*:=\mathsf M(0)=\mathbf f_*\mathbf f_*^T.
\label{eq:alena-critical-residue-master}
\end{equation}
Thus a simple single-soft KKT fold has an exact pole coefficient of rank at most one, and of rank one whenever the retained panel has nonzero soft coupling, $\mathbf f_*\neq0$. If $\mathsf M$ is $C^2$ in $\chi_Q$, $B_2^\chi:=\sup\|\partial_{\chi_Q}^2\mathsf M\|_{\rm op}$, and $\|\widehat{\mathsf M}_i-\mathsf M(\chi_i)\|_{\rm op}\leq\varepsilon_i$ at two stable-side nodes, then
\begin{equation}
\widehat{\mathsf R}_{12}^{[1]}:=\frac{\chi_2\widehat{\mathsf M}_1-\chi_1\widehat{\mathsf M}_2}{\chi_2-\chi_1},\qquad
\|\widehat{\mathsf R}_{12}^{[1]}-\mathsf R_*\|_{\rm op}\leq E_{12}:=\frac{|\chi_2|\varepsilon_1+|\chi_1|\varepsilon_2}{|\chi_2-\chi_1|}+\frac{B_2^\chi}{2}|\chi_1\chi_2|.
\label{eq:alena-two-box-total-error}
\end{equation}
Hence $\sigma_2(\widehat{\mathsf R}_{12}^{[1]})\leq E_{12}$ is a necessary rank-one check. A frame-complete target recovers the same critical tangent through its left inverse, while independent node pairs give a no-switching checksum. A proposed typed mediation $Y=F(X)$ is tested on the same critical line by comparing its reconstructed target tangent with $DF(X_*)$ applied to the source tangent within the certified residue error. For the explicit selector this stable-side route is retained as an independent check of the direct branch-jet export below rather than as its primary computation.
\end{theorem}

\begin{proposition}[Dual-Hessian chart of the selected first loss]
\label{prop:alena-scalar-jacobi-crossing}
On a nonhomogeneous fixed-phase block with $Dp_\Lambda=0$, put $R=R_{\rm vort}$, $g=(dR)^\sharp$, $q=\|g\|>0$, $e=g/q$, and write $D^2R=\left(\begin{smallmatrix}s_{00}&b^\dagger\\ b&S_\perp\end{smallmatrix}\right)$ relative to $\mathbb Re\oplus\ker dR$. Equation~\eqref{eq:alena-scalar-legendre-hessian} gives $A(C_R)=C_Rq^2P_e-\mu D^2R$. If $S_\perp\prec0$, its simple crossing is at $C_R=\mu(s_{00}-b^\dagger S_\perp^{-1}b)/q^2$, with kernel $e-S_\perp^{-1}b$ and one-unit inertia change. The transnormal residual $\|[P_e,D^2R]\|_{\rm HS}/\sqrt2$ vanishes exactly when $q^2$ is locally a function of $R$. Writing $\mathscr B(\Phi)=\|\nabla\Phi\|^2$, an oriented branch leaving a stable anchor then has its first scalar loss at
\begin{equation} \mathscr B_\Phi(\Phi_*)=0,\qquad \sigma_\Phi\mathscr B_{\Phi\Phi}(\Phi_*)<0. \label{eq:alena-oriented-first-loss} \end{equation}
The condition is invariant under $\Phi\mapsto g\Phi+c$, $g>0$; a uniform negative bound on $\sigma_\Phi\mathscr B_{\Phi\Phi}$ gives the usual one-dimensional root-error bound. Under identification of the dual-Hessian and constrained VH lines, this is a local chart of the fold selected by Theorem~\ref{thm:alena-protected-seed-to-fold}; phase relaxation is handled upstream.
\end{proposition}

The local Crandall-Rabinowitz description \cite{CrandallRabinowitz1971} is used only as a chart of the already selected simple fold. Let $K$ be the soft curvature, $\Delta\mathcal E$ the reduced energy gap, and suppose the same branch breaks the protected dark port as $\Gamma_{\rm dark}=\kappa_n|a|^{2n}+o(|a|^{2n})$. The slope-free closure and, after the same-arrow outgoing incidence, the fold/dark exponents are
\begin{equation}
\mathcal U_n:=\frac{\Gamma_{\rm dark}K^n}{\kappa_n|\Delta\mathcal E|^n}\longrightarrow
\begin{cases}6^n,&\beta_Q\neq0,\\8^n,&\beta_Q=0\ \text{by exact symmetry},\end{cases}
\qquad
\Gamma_{\rm dark}\sim|Q-Q_*|^n\sim|E-E_*|^{2n}\sim\Delta\mathscr G_{\rm fold}^{2n/3}.
\label{eq:alena-universal-branch-dark-closure}
\end{equation}
These are dependent consequences of the selected fold and outgoing channel, not additional selectors.

\begin{theorem}[Protected one-coordinate reduction]
\label{thm:alena-protected-one-coordinate-reduction}
On the constructed branch of Theorem~\ref{thm:alena-protected-seed-to-fold}, exact axial $SO(2)$ invariance, the native VH graph, charge-VH incidence, invertible eliminated directions and the declared sector gaps make all gapped tomographic variables there locally zero by symmetry or smooth functions of one longitudinal coordinate $x$. Thus $\mathcal P_{\rm AC,red}=\{V_{\rm eff}(x),L_{\rm eff}(x),H_{\rm jet}(x),\nabla_{\rm root}(x)\}$ and the exact zero of $\chi_Q$ is the first protected sector terminal; nonsinglet threshold completion remains separate.
\end{theorem}
The classical positive scale orbit multiplies $Q$, the unnormalized Hessian and $\chi_Q$ by one common factor but leaves the branch location, $s_*$, $u_*$ and normalized response geometry unchanged. Thus the exact fold cannot determine the absolute action scale. Common-root provenance along the full branch, the typed finite-active/outgoing response identification, the $q=5$ gap, nonsinglet thresholds and scattering remain downstream no-retune tests; intercell nesting is retained separately as a tower-universality question.

\subsection{Microscopic attachment: finite port, transport and response}
\label{subsec:alena-microscopic-normal-family}

The constructive Parent realization precedes microscopic attachment. The charge-Hamiltonian coordinate and microscopic spectral energy remain separately typed; the Dolbeault selector classifies Higgs data but does not by itself identify a Callias/Riesz kernel. Riesz/Fredholm conventions follow \cite{Kato1995}, \cite{Callias1978}, and \cite{Anghel1993Callias}; boundary/open-index comparisons use \cite{AtiyahPatodiSinger1975I} and \cite{Shi2017APSCallias}. Berry-Kato transport follows \cite{Berry1984GeometricPhase}, \cite{WilczekZee1984NonAbelianBerry}, \cite{BismutFreed1986EllipticFamiliesII}, and \cite{Quillen1985Superconnections}.

Let $L^{\rm ref}$ and $L^{\rm act}$ be the frozen reference and active operators after common-domain/common-kernel transport, with common kernel projector $P$, $Q={\bf1}-P$, and $E=L^{\rm act}-L^{\rm ref}$. For the frozen graph $\mathscr D_0$, the first attachment map is
\begin{equation} \mathscr F_E=(dE)_0P+E_0Q\mathscr D_0. \label{eq:alena-effective-attachment-source}\end{equation}
Exact first attachment neutrality is $\mathscr F_E=0$; at seed match $E_0=0$ it gives $(dE)_0P=0$ and, for a self-adjoint gapped family, $\dot P(0)=0$. The first attachment-driven Riesz motion therefore begins with the next independent source jet unless it is already present in $L^{\rm ref}$. On the seed-matched branch, $M=L_0^{\rm ref}Q$ gives the standard Moore-Penrose split $\mathscr J_{\rm att}=({\bf1}-MM^+)\mathscr F$ and $\delta\mathscr B=-M^+\mathscr F$: exact neutrality removes both cokernel obstruction and graph shift, while a nonzero residual localizes them.

\paragraph{Finite port.}
The normalized coisometric three-port used below is canonically supplied by the native $U$-submersion. If $\mathscr L_U:\mathcal H_\Phi^{\rm phys}\to C$ is the orthogonally projected native differential, then $G_U:=\mathscr L_U\mathscr L_U^\dagger\succ0$ and $\Gamma_U:=G_U^{-1/2}\mathscr L_U$ satisfy $\Gamma_U\Gamma_U^\dagger={\bf1}_C$; the existing submersion margin bounds $\|G_U^{-1/2}\|_{\rm op}$. Hence the moment port is taken as $\Gamma=\Gamma_U$, without a chosen local right inverse. The selector core does not determine the resulting port moments: quadratic Hessian completions invisible to the selector branch can change $S_1,\ldots,S_4$. These moments are therefore actual Parent export data rather than further selector invariants.

\begin{proposition}[Marked four-moment primitive-port certificate]
\label{prop:alena-marked-primitive-port-certificate}
\label{prop:alena-protected-port-self-dilation-converse}
Let $S_n=\Gamma K_C^n\Gamma^\dagger$, $n=0,\ldots,4$, with $S_0={\bf1}_C$, and define
\begin{equation} K=S_2-S_1^2,\qquad R=S_3-S_1S_2-S_2S_1+S_1^3,\qquad S_4^{\rm pred}=S_1S_3+KS_2+RS_1+RK^+R. \label{eq:alena-port-four-moment-closure} \end{equation}
Put $\Omega_4=S_4-S_4^{\rm pred}$. On the rank-two branch $K\succeq0$, $\rank K=2$, and $Rd=0$ for $d\in\ker K$. Specializing the matricial flat-extension criterion \cite{MadlerSchmudgen2024I} and \cite{MadlerSchmudgen2024II}, a finite self-adjoint realization through order four exists on this branch iff $\Omega_4\succeq0$ and $\Omega_4d=0$, and its minimal cyclic dimension is $5+\rank\Omega_4$. For literal Parent moments, in a unitary decomposition $K_C\simeq\left(\begin{smallmatrix}S_1&B\\B^\dagger&D\end{smallmatrix}\right)$ with $P_B:=B^\dagger K^+B$, one has $\Omega_4=BD({\bf1}-P_B)DB^\dagger\succeq0$. Hence exact-five is the stopping condition $\Omega_4=0$, with $\Tr\Omega_4$ a second-hidden-layer leakage residual.

On exact-five let $\mathcal H_K=\operatorname{Ran}K$, $B_0=K^{1/2}|_{\mathcal H_K}$ and $D_0=K^{+1/2}RK^{+1/2}|_{\mathcal H_K}$. Then
\begin{equation} \mathbb K_5^{\rm can}=\begin{pmatrix}S_1&B_0\\B_0^\dagger&D_0\end{pmatrix}\quad\text{on }C\oplus\mathcal H_K \label{eq:alena-port-canonical-minimal-realization} \end{equation}
reproduces $S_0,\ldots,S_4$ and is unique up to ${\bf1}_C\oplus U$ on the hidden two-space. If the port is marked as $C=\Sym^2W$ in the basis $(s_1^2,\sqrt2s_1s_2,s_2^2)$, primitivity is equivalent to $K_{02}=0$, $K_{01}=K_{12}$ and $2K_{11}=K_{00}+K_{22}$. On the nonzero rank-two branch these relations reconstruct $\alpha=K_{00}+K_{22}$ and the projective spinor $s$ by $|a|^2=K_{00}/\alpha$, $|b|^2=K_{22}/\alpha$, $a\bar b=\sqrt2K_{01}/\alpha$. With $X=\sqrt\alpha\,\jmath_s$, $U_s=K^{+1/2}X:W\to\mathcal H_K$ is unitary and $D_W=X^+R(X^+)^\dagger$, so
\begin{equation} \mathbb K_5^{\rm prim}=\begin{pmatrix}S_1&X\\X^\dagger&D_W\end{pmatrix}\quad\text{on }C\oplus W. \label{eq:alena-port-hidden-W-identification} \end{equation}
Thus the hidden two-space is the already reconstructed $W$. A direct zeroth-order equivariant injection of half-spin $W$ into the preattachment tensorial $SO(3)$ sector is excluded by the $SU(2)$ center.

On the same primitive stratum, $Z_\square:=\Tr(P_+S_1P_0RP_+)$ is the Hilbert-Schmidt pairing of the direct and hidden bright transports; when nonzero its phase is the intrinsic finite-core Bargmann/Wilson holonomy, while identification with the norm-seven orientation remains the marked-loop gate of Proposition~\ref{prop:alena-primitive-235-holonomy}. Point and first-jet primitivity remain separate. With $D_s=\nabla_sK$, $\mathcal L_{\rm mark}$ the linear primitive slice, $\mathsf P_{\rm mark}$ its Hilbert-Schmidt projector, and $E_d=\mathsf P_{\rm mark}(dd^\dagger)$,
\begin{equation} \operatorname{dist}_{\rm HS}(H,T_K\mathcal C_{\rm prim})^2=\|H-\mathsf P_{\rm mark}H\|_{\rm HS}^2+\frac65\langle E_d,H\rangle_{\rm HS}^2. \label{eq:alena-primitive-point-jet-projector} \end{equation}
For literal Gram jets $d^\dagger D_sd=0$, so G1-JET reduces to membership of $D_s$ in $\mathcal L_{\rm mark}$; the nearest compatible projection is nonexpansive and is used only as a one-sided error certificate.
\end{proposition}

\paragraph{Finite response and transport gates.}
\begin{proposition}[Finite-active Schur-Herglotz realization]
\label{prop:alena-passive-weyl-response}
\label{prop:alena-protected-channel-gap-states}
\label{prop:bl-finite-active-inverse-spectral-closure}
For $\mathbb K(z)=\left(\begin{smallmatrix}A-zZ_L&B\\B^\dagger&D-zZ_H\end{smallmatrix}\right)$ self-adjoint on the real axis with $Z_L,Z_H\succ0$, put $\mathsf K=Z_H^{-1/2}DZ_H^{-1/2}$ and $\mathsf C=BZ_H^{-1/2}$. The negative low Schur complement gives the finite-active matrix Herglotz-Nevanlinna response \cite{GesztesyTsekanovskii2000Herglotz}
\begin{equation}
\mathcal W(z)=-A+zZ_L+\mathsf C(\mathsf K-z{\bf1})^{-1}\mathsf C^\dagger=-A+zZ_L+\sum_a\frac{R_a}{\lambda_a-z},\qquad R_a\succeq0.
\label{eq:alena-weyl-schur-response}
\end{equation}
For $\mathsf K\succ0$, the positive moments $M_n(\alpha)=\mathsf C(\mathsf K+\alpha{\bf1})^{-(n+1)}\mathsf C^\dagger$ have block-Hankel rank equal to the reachable cyclic dimension, hence to the minimal finite realization degree by the moment problem \cite{KreinNudelman1977MomentProblem}. The finite matrix-moment viewpoint is compared with \cite{MadlerSchmudgen2024I} and \cite{MadlerSchmudgen2024II}; block-Jacobi and matrix Titchmarsh-Weyl formulations are represented by \cite{MoszynskiSwiderski2025}. A full outgoing completion may carry an additional positive matrix-Herglotz remainder and need not be rational. On the common resolvent set the compressed Green response is $\mathcal G(z)=-\mathcal W(z)^{-1}$; an exported compressed resolvent is compared with $\mathcal G$, whereas boundary-Weyl output is compared with $\mathcal W$ after the affine boundary gauge below has been fixed. Moment recurrences, scalar pole-zero interlacing and positive-Hankel stability are standard consequences of the same finite realization theory.
\end{proposition}
Boundary-triple conventions follow \cite{Grubb1996FunctionalCalculus} and \cite{BoossBavnbekWojciechowski1993DiracBoundary}. Under $\Gamma_0\mapsto C\Gamma_0$ and $\Gamma_1\mapsto C^{-*}\Gamma_1$, one has $M\mapsto C^{-*}MC^{-1}$. On the marked irreducible port $C=\Sym^2W$, the residual equivariant affine class is $M\mapsto r^{-2}(M+b{\bf1})$; a nonconstant same-object scalar lock of Proposition~\ref{prop:alena-projective-one-pole} at two energies forces $r=1$ and $b=0$. Only this affine gauge is fixed; absolute action scale and transport connection remain open, while whitened data remain shape tests.\
~\\
By \eqref{eq:pgl-source-generated-finite-core}, $A\in\End(C)$ and $D\in\End(W)$ already supply the Wilson-breaking slots, with coefficients remaining Parent data. For relative charges $(1,0,-1;1,0)$, nonzero inter-sector blocks must form a connected graph; one edge is insufficient.

\begin{proposition}[Projective one-pole and arithmetic-shell certificate]
\label{prop:alena-projective-one-pole}
For a marked unit scalar channel $\ell$, subtract the affine Herglotz part and put $m_{\rm act}(E):=a_0+\ell^\dagger(\mathcal W(E)+A-EZ_L)\ell$. One finite visible pole gives $m_{\rm act}(E)=a+r/(\lambda-E)$, $r>0$. Its standard projective characterizations are zero Schwarzian, rank-one Loewner/Pick kernels and preserved cross-ratios. The same-object Cayley lock $\mathcal C(m_{\rm act})=\mathcal C(\sqrt3\,q_{\rm fam})$, with $\mathcal C(x)=(x+i)/(x-i)$, fixes $m_{\rm act}=\sqrt3\,q_{\rm fam}$.\\
~\\
On one trace-one arithmetic orientation, $q_m=2m+1$ gives
\begin{equation}
E_m=\lambda-\frac{r}{\sqrt3(2m+1)-a},\qquad
\frac1{\lambda-E_{m+1}}-\frac1{\lambda-E_m}=\frac{2\sqrt3}{r}.
\label{eq:alena-reciprocal-shell-lattice}
\end{equation}
Three consecutive shells identify $(\lambda,r,a)$ uniquely whenever the interpolation denominator is nonzero. The fourth shell is a no-refit holdout and the cross-ratio of four consecutive same-orientation values is $4/3$; in the odd-cell sequence this makes $q=9$ the first blind clean-prime holdout, while $q=11$ is the first composite stress test. For $q''\neq0$, the local jet $(q,q',q'')$ identifies the same one-pole law and provides a redundant same-object check.
\end{proposition}

\begin{proposition}[Product-commutator centrality certificate]
\label{prop:alena-heisenberg-resolvent-centrality}
Let the finite active block carry the fundamental $SU(3)$ and $SU(2)$ actions and the marked family Weyl pair $S,Z$, with $\Tr(T_a^\dagger T_a)=1/2$, and let $\mathbb E_{\rm prod}$ be the trace-preserving conditional expectation \cite{GaoJungeLaRacuente2020SubalgebraEntropy} onto the product commutant. For $\mathcal C_{\rm prod}(X):=\sum_a\|[X,T_a]\|_{\rm HS}^2+\|[X,S]\|_{\rm HS}^2+\|[X,Z]\|_{\rm HS}^2$, the factor commutator Laplacians have nonzero spectra $\{3\}$ on $M_3$, $\{2\}$ on $M_2$, and $\{3,6\}$ for the marked family Weyl pair; hence the tensor-sum Laplacian has smallest and largest positive eigenvalues $2$ and $11$, so $\mathcal C_{\rm prod}(X)/11\leq\|X-\mathbb E_{\rm prod}X\|_{\rm HS}^2\leq\mathcal C_{\rm prod}(X)/2$ off the scalar commutant, and for $R_H(z)=(K_H-z)^{-1}$ with $g_H(z)=\operatorname{dist}(z,\sigma(K_H))>0$ one has $\|R_H-\mathbb E_{\rm prod}R_H\|_{\rm HS}\leq\sqrt{\mathcal C_{\rm prod}(K_H)/2}/g_H(z)^2$. A typed Schur edge multiplies this bound by at most its squared operator norm. Exact product equivariance makes the denominator scalar on the protected product factor; approximate equivariance gives the corresponding propagation budget.
\end{proposition}

\paragraph{Outgoing channel.}
\begin{proposition}[Primitive port continuum and moment protection]
\label{prop:alena-primitive-dark-port}
\label{prop:alena-graph-distance-moment-protection}
Let the continuum couple behind $W$ through the primitive arrow $X:W\to C$. For $\Sigma_C^{\rm dyn}(z)=XM(z)X^\dagger$, one has
\begin{equation} L_{\rm dark}\subseteq\ker\Sigma_C^{\rm dyn}(z)\qquad\text{for every }z. \label{eq:alena-primitive-dark-port-full-self-energy} \end{equation}
Thus same-arrow absorption has rank at most two and annihilates the frozen dark line. Its finite-block breaking is measured by $\mathcal R_{\rm dark}=\|({\bf1}-P_D)A_{\rm eff}P_D\|_{\rm HS}$; exact BIC status additionally requires the outgoing/domain condition. For a moving rank-one $P_D$, standard projector transport \cite{Kato1995} gives $H_{\rm cd}=i\hbar[\dot P_D,P_D]$. Writing $D_s=\nabla_sK$, $T=\Tr K$, and $V=\|D_s-(\Tr D_s){\bf1}/3\|_{\rm HS}^2$, the primitive jet gives $\mathcal S_D=\|\nabla_sP_D\|_{\rm HS}^2=9V/(2T^2)-(\Tr D_s)^2/T^2$. The actual gate is $\mathcal R_{\rm cd}=\|({\bf1}-P_D)(H_{\rm dyn}-H_{\rm cd})P_D\|_{\rm HS}$. If the same fixed-domain Wilson block is identified with the norm-seven seed, exact frozen-line protection is excluded by $\mathcal R_{\rm dark}\geq4\sqrt{33}g/49$; after the moving-dark same-object incidence this becomes the conditional bound $|\dot s|\sqrt{\mathcal S_D}\geq4\sqrt{66}g/(49\hbar)$.

If this mismatch is the literal outgoing leakage, write $\delta b=|b\rangle\langle d|$, $\|b\|^2=\mathcal R_{\rm cd}^2$, and $m_{\rm leak}(z)=\langle\widehat b,M_{\rm out}(z)\widehat b\rangle$. Then
\begin{equation} \Gamma(E_*)=2\mathcal R_{\rm cd}^2\Im m_{\rm leak}(E_*+i0). \label{eq:alena-dark-scalar-weyl-width} \end{equation}
Only the cyclic reducing subspace generated by $(K_{\rm out},b)$ contributes. The affine-normalized scalar channel removes the additive boundary shift and multiplicative shape scale; absolute lifetime still requires the literal leakage amplitude and, unless a unit-resolvent export fixes it, the boundary normalization. The regular-density law is $\Gamma=O(\varepsilon^{2n})$; for threshold density $\Im m_{\rm leak}(E_{\rm thr}+x+i0)\sim|x|^\alpha$ and displacement $O(\varepsilon^p)$ the exponent is $2n+p\alpha$. Counterdiabatic performance bounds provide a recent comparison \cite{FinzgarEtAl2025Counterdiabatic}. A stationary Parent does not determine $\dot s$. A tail beginning at graph distance $L$ leaves $S_n$ unchanged for $n<2L$, with $S_{2L}$ the first moment that can change.
\end{proposition}
Nahm data remain comparisons: the standard correspondence is \cite{Hitchin1983} in the ADHMN convention \cite{Nakajima1993}, with axial/rational representatives \cite{HoughtonSutcliffe1996}, \cite{IoannidouSutcliffe1999}, and \cite{Donaldson1984}; cyclic spectral data are represented by \cite{BradenDAvanzoEnolski2011}, \cite{BradenDisneyHogg2023}, and \cite{BradenEnolski2018}.

\paragraph{Centered response.}
For the centered low cell the retained entries are
\begin{equation}
A=A_0-\frac32u,\qquad B=B_0,\qquad C=C_0+u.
\label{eq:bl-neutral-cell-corrected-entries}
\end{equation}

If $PT_\ell P=0$ and $QK_0Q$ has gap $g_{\rm CSch}$, standard Feshbach reduction \cite{Kato1995} gives, for $\ell\geq3$, the second-order feedback $-\varepsilon^2PT_\ell Q(QK_0Q)^{-1}QT_\ell P$ with remainder $O(|\varepsilon|^3\|T_\ell\|^3/g_{\rm CSch}^2)$.
Under exact finite-shadow covariance, parity closure and an isolated finite-core primitive projection, the low module factorizes into the primitive channel and the retained central family factor. A central perturbation below half the isolating gap preserves the product rank; the one-dimensional parity factor and three-dimensional central family factor therefore give rank-three low clusters.

\begin{proposition}[Canonical centered-Riesz recovery and minimal response attachment]
\label{prop:alena-centered-riesz-recovery}
Let $a_{\rm cen}(R)=2\Tr_C R/15-\Tr_W R/5$ be the $\mu_K=P_C-3P_W/2$ coefficient of a positive Parent response path. The selector core fixes the terminal but does not determine a unique normalized response: positivity, unitary covariance, scale invariance and trace normalization still allow inequivalent centered functionals, e.g. the family $5(H+\tau\Tr(H){\bf1}/5)^{-1}/\Tr[(H+\tau\Tr(H){\bf1}/5)^{-1}]$ for $H=\diag(1,2,4;3,5)$. Once such a typed centered-response path is supplied, the marked orientation and the normalization of \eqref{eq:bl-neutral-cell-corrected-entries} give
\begin{equation} u(\Omega)-u(\Omega_a)=\frac23[\widetilde a_{\rm cen}(\Omega)-\widetilde a_{\rm cen}(\Omega_a)]. \label{eq:alena-parent-to-u-quadrature} \end{equation}
At the selected terminal,
\begin{equation} u_{\rm Parent}=u_*=\frac23[\widetilde a_{\rm cen}(\Omega_*)-\widetilde a_{\rm cen}(\Omega_{\rm str})],\qquad \chi_Q(\Phi_*)=0. \label{eq:alena-parent-selector-pushforward} \end{equation}
For the explicit selector, after an actual primitive point+jet pass and the same-object centered-response incidence, let $X_C(\Phi)$ denote the typed native $C$-target entering the terminal jet, $D_{s,*}:=\nabla_sK_*$, and $a_{\rm str}:=\widetilde a_{\rm cen}(\Omega_{\rm str})$. The invariant terminal scalar is
\begin{equation} A_*^{\rm jet}:=\frac32\left(\frac{q_0}{e_0}\right)^2\left[\frac4{45}\|\nabla_sX_C|_*\|^2-\frac{\|D_{s,*}-(\Tr D_{s,*}){\bf1}/3\|_{\rm HS}^2}{10\Tr K_*}\right],\qquad u_* =\frac23(A_*^{\rm jet}-a_{\rm str}). \label{eq:alena-parent-branch-jet-u} \end{equation}
The stable-side residue remains an independent checksum.

On exact-five the microscopic operator is \eqref{eq:alena-port-hidden-W-identification}. If the physical response is the literal static linear susceptibility in these same marked coordinates, its raw functional is $(\mathbb K_5^{\rm prim})^{-1}$. Let $r_*,r_{\rm str}$ be its centered block traces at endpoint and anchor, and let $\sigma_{\rm cen}=\pm1$ be the upstream orientation fixed by the marked centered line. In the preregistered class $\mathscr R_P(\Phi)=g(\mathbb K_5^{\rm prim}(\Phi))^{-1}+c(\Phi){\bf1}$ with one common $g>0$, the terminal jet fixes $g=A_*^{\rm jet}/r_*$ provided $r_*\neq0$ and $A_*^{\rm jet}/r_*>0$, and
\begin{equation}
u_{\rm Parent}=\frac23\sigma_{\rm cen}A_*^{\rm jet}\left(1-\frac{r_{\rm str}}{r_*}\right).
\label{eq:alena-parent-static-u}
\end{equation}
This is a conditional blind Parent output and requires a same-object static/common-gain certificate without low-energy longitudinal data. Otherwise one predeclared chart datum selects the realized longitudinal member and is not rescored.
\end{proposition}
Common-root provenance is retained through one common analytic characteristic equation with a marked simple root fixed at the anchor and continued while its root gap remains nonzero. The scalar characteristic lock is
\begin{equation}
\mathcal R_{JS}:=C_R-C_*(q_E)=0,
\label{eq:alena-family-scalar-characteristic-lock}
\end{equation}
The terminal residual and its discrete closure are
\begin{equation}
\mathcal R_{\rm term}^{(n)}=\left\|r_AV-\frac2nR\right\|^2=r_A^2\mathcal R_{\rm sc}+\left(r_Ac_*-\frac2n\right)^2,
\label{eq:bl-neutral-terminal-residual}
\end{equation}
\begin{equation}
|q_{23}|=\frac{r_An g_{\rm rel}}{2\sqrt2},\qquad \widehat\Xi_{23}=\frac{\sqrt3\,r_An g_{\rm rel}}{4\sqrt2},\qquad \mathcal R_{\rm term}^{(n)}=0\Longleftrightarrow|q_{23}|=\frac1{\sqrt2}.
\label{eq:bl-arc-takagi-discrete-closure}
\end{equation}
\begin{theorem}[PLSR elementary attachment and terminal closure]
\label{thm:bl-elementary-terminal-closure}
Assume the integral incidence, finite-parent regularity, and marked common ambient of this section. If K1-K3 recover their corresponding protected data by exact PLSR, then $|p|=2$, selecting the two-torsion cokernel in \eqref{eq:alena-c3-real-integral-incidence}, while $r_A=1$, $\mathcal R_{\rm sc}=0$, and $g_{\rm rel}=1$. Consequently \eqref{eq:bl-neutral-terminal-residual} vanishes and \eqref{eq:bl-arc-takagi-discrete-closure} gives $|q_{23}|=1/\sqrt2$ and $\widehat\Xi_{23}^2=3/8$. On the scalar-relay branch, $\bigl||q_{23}|-1/\sqrt2\bigr|\leq\sqrt{\mathcal R_{\rm term}^{(2)}/2}$.
\end{theorem}
\begin{corollary}[Terminal-Hodge incidence]
\label{cor:bl-plsr-source-hodge-coherence}
For the marked source-Hodge folding, let $\delta_q$ be the normalized common-to-relative third-jet defect. The source geometry gives $q_*=4+\kappa_g^2$, $\delta_q=9\kappa_g^2+\dot\kappa_g^2$, and $q_*r_{\rm src}^2=8/3$. On the exact terminal branch define $\mathcal R_{\rm TH}:=\chi_H^2-\widehat\Xi_{23}^{\,2}$. Using Theorem~\ref{thm:alena-dual-veronese-source-jet} and Theorem~\ref{thm:bl-elementary-terminal-closure},
\begin{equation}
\mathcal R_{\rm TH}=\frac38\kappa_g^2,\qquad
q_*-4=\frac83\mathcal R_{\rm TH},\qquad
r_{\rm src}^2=\frac{2}{3+2\mathcal R_{\rm TH}},\qquad
\delta_q=24\mathcal R_{\rm TH}+\dot\kappa_g^2.
\label{eq:alena-terminal-hodge-cross-residual}
\end{equation}
Thus $\mathcal R_{\rm TH}=0$ is the single terminal-Hodge incidence that fixes the geodesic source point and $q_*=4$; $\delta_q=0$ remains an independent pointwise covariant-prolongation test. Along one smooth branch on which $\mathcal R_{\rm TH}$ vanishes identically, the same relation forces the prolongation defect to vanish identically. Terminal closure alone does not imply this incidence.
\end{corollary}

\section{Scale-Free Low-Energy Consequences}
\label{sec:quantitative-branch-tests}

The quantities in this section are evaluated only after the structural packet, Parent realization and microscopic attachment have been fixed. No quantity below is used to reconstruct $\mathfrak P_\Gamma$ or to select a Parent term. The section therefore reads one previously constructed reduced manifold: calibration fixes one longitudinal coordinate, independently sourced transverse relations are retained as no-retune tests, and dependent readings and downstream completion data are identified separately.

\subsection{Schur transfer and the reduced low-energy operator}
\label{subsec:filtered-schur-transfer}

Schur-Feshbach reduction is applied only after the structural packet, Parent realization and microscopic attachment have been fixed. Homological/filtered descent is represented by \cite{DoubekJurcoPulmann2019HPL}; differential-index conventions follow \cite{FreedLott2010DifferentialKIndex}.

\begin{lemma}[Schur transfer and filtered walks]
\label{lem:fsr-penalty-descent}
\label{lem:bl-filtered-graph-walk-calculus}
The reduced Hessian and isolated low operator are the Schur/Feshbach complements already used in \eqref{eq:alena-fixed-q-schur-update}. Their Neumann expansion is the standard homological-perturbation expansion \cite{DoubekJurcoPulmann2019HPL}. For a filtered finite matrix, entries are sums over walks; the leading valuation is the minimum walk valuation unless the minimizing coefficients cancel, and after vertex rephasing the independent phases are cycle-space phases. For the marked zero-diagonal Eisenstein companion splitting the scalar reading is $N_E(\xi)/\Tr\xi$; an orthonormal real presentation gives twice this value, so the splitting remains part of the protected datum.
\end{lemma}
\begin{proposition}[Gapped Schur stability]
\label{prop:bl-gapped-schur-stability}
Let $K=\left(\begin{smallmatrix}A&B\\B^\dagger&D\end{smallmatrix}\right)$, $S=A-BD^{-1}B^\dagger$, and assume $\sigma_{\min}(D)\geq g>0$. Put $\widetilde S=(A+\delta A)-(B+\delta B)(D+\delta D)^{-1}(B+\delta B)^\dagger$, and let all norms be operator norms. If $\|\delta A\|,\|\delta B\|,\|\delta D\|\leq\varepsilon<g$ and $M=\|B\|$, then
\begin{equation}
\|\widetilde S-S\|\leq\varepsilon\left[1+\frac{M+\varepsilon}{g-\varepsilon}+\frac{M(M+\varepsilon)}{g(g-\varepsilon)}+\frac{M}{g}\right].
\label{eq:bl-gapped-schur-stability}
\end{equation}
Consequently an independently established operator-readout bound $\varepsilon\leq c\sqrt{\delta_{\rm fid}}$ gives $O(\!\sqrt{\delta_{\rm fid}})$ stability of the reduced operator while the heavy gap remains open.
\end{proposition}

The reduced outputs form one map $Y:(u,\sigma_7^{\rm ar})\mapsto Y(u,\sigma_7^{\rm ar})$. Affine, flavor and mass relations are coordinate readings; normalization, microscopic pole locations and running remain completion fibers. The common affine cell contains five scale-free observables with four independent local relations, so one predeclared chart observable fixes the longitudinal coordinate. The spectral-action comparison \cite{ChamseddineConnes1997SpectralAction} and current reference values \cite{PDG2026RPP} are used only after this analytic reduction.

\subsection{Charged-sector and CKM consequences}
\label{subsec:family-transfer-all-order}

Modular flavor comparisons are represented by \cite{Feruglio2019ModularNeutrino}; generalized CP by \cite{NovichkovPenedoPetcovTitov2019ModularCP}, with stabilization near $\tau=\omega$ in \cite{NovichkovPenedoPetcov2022ModulusStabilisation}. Quark-model comparison is represented by \cite{PetcovTanimoto2023TauOmega}. Recent hierarchy, realization and partially broken CP studies are represented by \cite{ChenLiLiuRatz2025ModularHierarchies}, \cite{ArriagaOsanteEtAl2025MultipleRealizations}, and \cite{HigakiKobayashiNasuOtsuka2024CP}.

\paragraph{Mass shells and CKM reconstruction.}
The charged sectors are described by normalized amplitudes $r_{f,i}=\sqrt{m_{f,i}}/\sum_j\sqrt{m_{f,j}}$. Their determinant, cubic, Cartan-angle and conic forms are coordinate charts of one fixed-$Q$ $A_2$ circle, not independent constraints.

\begin{proposition}[Fixed-$Q$ $A_2$ shell, determinant completion, and Rees flow]
\label{prop:bl-direct-hodge-one-scalar-shape}
\label{prop:bl-fixed-q-completion-circle}
Let $r_i$ be the eigenvalues of $H_f/\Tr H_f$. For $1/2<Q_f<1$, the constraints $\sum_ir_i=1$, $\sum_ir_i^2=Q_f$ define the intersection of the affine $A_2$ plane with a circle. Up to the $D_3$ action it is parametrized by
\begin{equation}
r_k=\frac13+A_Q\cos\left(\theta+\frac{2\pi k}{3}\right),
\qquad
A_Q=\sqrt{\frac23\left(Q_f-\frac13\right)},
\qquad k=0,1,2.
\label{eq:bl-fixed-q-a2-circle}
\end{equation}
Positivity and ordering select the physical arc. The symmetric determinant $\pi_f=r_1r_2r_3$ reconstructs the unordered amplitudes through
\begin{equation}
z^3-z^2+\frac{1-Q_f}{2}z-\pi_f=0,
\label{eq:bl-direct-hodge-sector-cubic}
\end{equation}
with three distinct positive roots for $0<\pi_f<\pi_+(Q_f)$, $\pi_+(Q)=[5-9Q+\sqrt2(3Q-1)^{3/2}]/54$. The determinant, $\Tr X_f^3$ and Cartan angle are coordinate charts of the same one-dimensional orbit; equivalently, for $a=\sqrt{m_1/m_3}$ and $b=\sqrt{m_2/m_3}$ the ordered conic obeys $(1-Q_f)(1+a^2+b^2)=2Q_f(a+b+ab)$. Equal-amplitude collisions are chart singularities, while $\pi_f=0$ is the distinct rank-loss boundary. Projecting $H_R=(3/2)H_F-(1/2)K_F$ to the shell tangent gives
\begin{equation}
\frac{d\theta}{dt}=\frac{3\cos\theta-\sin\theta/\sqrt3}{\sqrt{2(Q_f-1/3)}}.
\label{eq:bl-fixed-q-exact-rees-flow}
\end{equation}
A linear $S_3$-equivariant endomorphism preserving the affine plane and $Q_f$ restricts to $\pm{\bf1}$ on the centered plane; nontrivial shell motion therefore requires a nonlinear Parent response or a separately retained sector gain. A universal raw gain remains a Parent test rather than a consequence of the Rees grading.
\end{proposition}

\begin{proposition}[Determinant-fiber degree and jet hierarchy]
\label{prop:bl-determinant-fiber-degree-hierarchy}
At fixed $I_2=\Tr X_f^2$, Cayley-Hamilton reduces every conjugation-invariant polynomial Parent selector to a polynomial in $I_3=\Tr X_f^3$. Degree $d\leq5$ gives no generic strict interior well, $6\leq d\leq8$ at most one quadratic well, $d\geq9$ first permits an interior barrier, and $d\geq12$ first permits two strict interior minima; in general $M$ strict interior minima require $d\geq6M$. Through degree seven the fiber curvature is common after the locations are fixed; $I_2I_3^2$ at degree eight first releases this stiffness and degree nine first permits a cubic fiber jet. Writing the same dependence as $h(Q_f)I_3$ does not introduce a separately typed sector marker.
\end{proposition}
\begin{proposition}[Scalar finite quark endpoint]
\label{prop:bl-scalar-finite-quark-endpoint}
Let $\widehat H_f=H_f/\Tr H_f$, $f=u,d$, have simple ordered eigenvalues $r_{f,i}$ fixed by \eqref{eq:bl-direct-hodge-sector-cubic}. Put $P_i^f$ for the spectral projectors, $B_{ij}=\Tr(P_i^uP_j^d)$, $M_{ab}=\Tr(\widehat H_u^a\widehat H_d^b)$, and $V_f=(r_{f,j}^{\,a})_{a=0,1,2;\,j=1,2,3}$. Then
\begin{equation} M=V_uBV_d^{\mathsf T},\qquad B=V_u^{-1}MV_d^{-\mathsf T}. \label{eq:bl-mixed-moment-overlap-reconstruction} \end{equation}
The four new mixed scalars $M_{11},M_{12},M_{21},M_{22}$ reconstruct $B$; the fixed one-sector moments enforce unit row and column sums, and the Jacobian on the doubly-stochastic interior is nonzero for simple spectra. Thus four continuous mixed scalars are generically locally minimal. For ordered eigenframes $B_{ij}=|V_{ij}|^2$. Its $3\times3$ unistochastic physicality condition and Jarlskog magnitude are
\begin{equation}
\mathscr U(B):=4B_{11}B_{12}B_{21}B_{22}-\bigl(B_{13}B_{23}-B_{11}B_{21}-B_{12}B_{22}\bigr)^2\geq0,\qquad |J_q|=\frac12\sqrt{\mathscr U(B)}.
\label{eq:bl-unistochastic-physicality}
\end{equation}
For $\mathscr U>0$ only the two conjugate orientations remain; at $\mathscr U=0$ the CP sign is not an independent datum.
\end{proposition}
The mixing conventions are standard \cite{Cabibbo1963}, \cite{KobayashiMaskawa1973}, \cite{Wolfenstein1983}, and \cite{Jarlskog1985}; the hierarchy comparison is \cite{FroggattNielsen1979Hierarchy}.

On the selected Weyl directions $X=Z$ and $J=\omega Z^2S$, a central-denominator Schur transfer with $A(t)=txX+t^2yJ/2$ and $G=g{\bf1}$ resums exactly to
\begin{equation}
\mathcal S(t)=\frac{A(t)-gA(t)^2+g^2\left[(tx)^3+(t^2y/2)^3\right]{\bf1}}{1+g^3\left[(tx)^3+(t^2y/2)^3\right]}.
\label{eq:bl-all-order-minimal-family-transfer}
\end{equation}
Hence $\mathcal S(t)\in\operatorname{span}\{{\bf1},S,Z,ZS^2,Z^2,Z^2S\}$; components along $S^2$, $ZS$, or $Z^2S^2$ are forbidden after minimization over the common marked normalizer. This exact support is a falsifier of the transfer law; a finite local jet alone does not establish it.

\begin{corollary}[Primitive contraction fixed point and Jacobi response]
\label{cor:bl-canonical-jacobi-second-response}
Assume that the source/fuzzy contraction, source-Hodge map, and family Jacobi response are typed as three readouts of the same normalized $V_1$ contraction. If $\gamma\in[0,1]$ is its normalized defect eigenvalue, survival/defect complementarity gives $c_2^2=1-\gamma$ and $c_3=\gamma$. The same marked Rees transfer gives $c_2^2=(1+\rho_{\rm RP}^2)/3$ and $c_3=\rho_{\rm RP}/3$. Together these relations force
\begin{equation}
\gamma=\frac13,\qquad \rho_{\rm RP}=1,\qquad c_2^2=\frac23,\qquad c_3=\frac13.
\label{eq:bl-primitive-contraction-fixed-point}
\end{equation}
Compatibility with the adjacent fuzzy tower then selects $p=2$; on the typed source-Hodge branch the same relay gives $q_*=4$, $\kappa_g=0$, $r_{\rm src}^2=2/3$ and $\chi_H^2=3/8$. The identity $r_{\rm src}^{-2}-\chi_H^2=9/8$ is retained as a cross-reading. Polarization is not fixed by complementarity; after the selected orientation also gives $\Delta_{\rm pol}=0$, the primitive complement satisfies $\Tr Y_\perp=\Tr(Y_\perp^2)=0$, $Y_\perp^3=(\det Y_\perp){\bf1}$, and $[X,Y_\perp]^\dagger[X,Y_\perp]=3a_{\rm RP}^4{\bf1}$. Complex conjugation gives the opposite orientation.
\end{corollary}
\paragraph{Filtered hierarchy and CP orientation.}
\begin{proposition}[$H_R$ bigrading, geodesic filtration, and Parent valuations]
\label{prop:bl-filtered-hermitian-denominator}
Assume a marked identification of the family $A_2$ frame carrying $H_R=J_3/3$ from \eqref{eq:pgl-family-fourier-master} with the positive response $A_2$ frame. In the marked order $h_1=4/3$, $h_2=1/3$, $h_3=-5/3$ give metric degrees $\nu_{12}=1$, $\nu_{23}=2$, and $\nu_{13}=3=\nu_{12}+\nu_{23}$. Relative to the common projective $C_3$ stabilizer, $(E_{12},E_{23},E_{13})$ have characters $(\omega,\omega,\omega^2)$, so the product $E_{12}E_{23}=E_{13}$ preserves the bigrading $(1,1)+(2,1)=(3,2)$ in $\mathbb Z_{\geq0}\times\mathbb Z_3$. Hence $F_pM_3:=\operatorname{span}\{E_{ij}:\nu_{ij}\leq p\}$ is a multiplicative metric filtration and the degree-three $13$ channel is generated at the graded level. The remaining dynamical family-frame gate is the actual Parent valuation
\begin{equation}
\operatorname{ord}(12,23,13)=(1,2,3)
\label{eq:bl-hr-rees-valuation-lock}
\end{equation}
in this same marked $H_R$ frame. The $C_3$ grading alone does not generate $(1,2,3)$; a microscopic charged-hierarchy reading requires same-object identification with gap-normalized leakage channels. Source-level character coefficients and the Rees coordinate remain separately typed until a same-object source-to-saturation map is constructed, so their perturbative powers are not identified.
Let the leading positive Hermitian germ be $H(t)=D+t(a_{12}E_{12}+\overline a_{12}E_{21})+t^2(a_{23}E_{23}+\overline a_{23}E_{32})+t^3(a_{13}E_{13}+\overline a_{13}E_{31})+\cdots$, with $D>0$ diagonal. The filtration fixes the order of the $13$ channel but not its direct coefficient: $a_{13}$ remains independent of the composite $a_{12}a_{23}$. The inverse has the same lower valuations by Lemma~\ref{lem:bl-filtered-graph-walk-calculus}. With all three leading edges present, the first oriented cycle has weighted degree six; a scalar Hermitian determinant sees its real part but cannot fix its orientation. On a tree stratum the cycle phase disappears while the order-three hierarchy remains. This is Froggatt-Nielsen-type power counting \cite{FroggattNielsen1979Hierarchy} with geodesic Rees degree replacing an additional flavor charge.
\end{proposition}
\begin{proposition}[Direct/composite Wilson quotient and leading family packet]
\label{prop:bl-gap-wilson-family-packet}
On the all-edge-nonzero germ of Proposition~\ref{prop:bl-filtered-hermitian-denominator}, vertex rephasing leaves one phase-bearing coordinate $\mathfrak W_\triangle:=a_{12}a_{23}\overline{a_{13}}/|a_{12}a_{23}|^2=\overline{a_{13}}/\overline{a_{12}a_{23}}$, with $\mathcal I_\triangle=|a_{12}a_{23}|^2\Im\mathfrak W_\triangle$. The quotient is invariant under vertex rephasing and regular rescaling of the Rees coordinate. If $a_{13}=\kappa a_{12}a_{23}$ with real $\kappa$, then $\mathcal I_\triangle=0$; for $a_{13}=\kappa a_{12}a_{23}+a_{13}^{\rm dir}$, only the direct/composite interference contributes to the imaginary invariant. Thus the Rees composition fixes the hierarchy order but does not generate CP by itself.
For the same leading relative generator, the same-object amplitude incidence is $|a_{23}|/|a_{12}|^2=c_2$ and $|a_{13}|/|a_{12}|^3=c_3$, so $|\mathfrak W_\triangle|=c_3/c_2$. After the primitive contraction fixed point \eqref{eq:bl-primitive-contraction-fixed-point} and the physical norm-seven gate \eqref{eq:alena-norm-seven-wilson-gate} have both passed, the leading complex target is
\begin{equation} \gamma=\frac13,\qquad \sqrt6\,\mathfrak W_\triangle=w_7,\qquad \lim_{\lambda\to0}\frac{J_q}{\lambda^6}=(1-\gamma)\Im\mathfrak W_\triangle=\sigma_7^{\rm phys}\frac{60\sqrt2}{343}. \label{eq:bl-gap-wilson-packet} \end{equation}
Hence $c_2^2=2/3$, $c_3=1/3$, the ratio $1/\sqrt6$, the CP orientation and the leading Jarlskog coefficient are dependent readings of one passed contraction-holonomy packet rather than independent fits.
\end{proposition}
With the three retained quark angles and reporting orientation $\sigma=\pm1$, put $s_{13}=\sin\theta_{13}$ and $c_{13}=\cos\theta_{13}$. Before an actual PASS of \eqref{eq:alena-norm-seven-wilson-gate}, $\sigma$ is only the binary structural reporting branch. The complex completion is
\begin{equation}
V_{\rm CKM}(u,\sigma)=R_{23}(\theta_{23})
\begin{pmatrix}c_{13}&0&s_{13}e^{-i\sigma\delta_7}\\0&1&0\\-s_{13}e^{i\sigma\delta_7}&0&c_{13}\end{pmatrix}R_{12}(\theta_{12}).
\label{eq:bl-ckm-complex-completion}
\end{equation}
The sign is fixed only by the physical marked-cycle incidence; the opposite sign is the conjugate orientation.

The CM-Real affine theorem of Sec.~\ref{subsec:pgl-equianharmonic-family} provides the zeroth-order neutral mixing frame. In its canonical triplet realization, marked complex conjugation exchanges the dual point/line $S_4$ extensions of the same $C_3$ and the two norm-seven orientations; triplet spectra and $|U|^2$ do not select between them. In the marked orthonormal basis of the two-dimensional Klein response plane, write $t=b/a$ for the projective point $[a:b]$. The outer-twisted antiunitary fixed locus is the exact circle $|t+2\sqrt2|=3$, with phase-invariant residual $\mathcal R_{\rm out}(t):=\bigl||t|^2+4\sqrt2\Re t-1\bigr|/[3(1+|t|^2)]$ and $\mathcal R_{\rm out}(0)=1/3$. Within this projective completion, exact outer-Real therefore requires branch exchange or the full two-channel action. Physical selection remains a same-carrier point/line-Real incidence. After that gate the nonnormal Klein and TBM overlap are fixed; full diagonalization is a checksum. Entanglement-based flavor selection remains an independent comparison \cite{ThalerTrifinopoulos2025FlavorEntanglement}.

For charged leptons the zero-correction $\mathbb Z_3$ Fourier branch gives the Koide shell. The original relation is \cite{Koide1983LeptonMassFormula}, with geometric and clock comparisons \cite{Foot1994KoideGeometry} and \cite{Zenczykowski2013KoideZ3}. Higher response enters only through the already fixed Schur hierarchy.

\subsection{Neutral relay and PMNS consequences}
\label{subsec:neutral-schur-complement-pfaffian}

For the complex-symmetric neutral block the standard seesaw/Schur reduction is used as in \cite{Senjanovic2005SeesawGUT}, with normalized shape separated from the overall mass scale. Since the protected weak-direct descendants preserve \eqref{eq:pgl-derived-bl-envelope}, a Majorana block requires a separately typed $\Delta(B-L)=2$ datum. Oscillation conventions and current regions follow \cite{Esteban2024NuFIT60} and \cite{NuFIT2025v61}.

\begin{equation}
K_\nu^{\rm eff}=A_L-D_\nu^T(M_RR_N)^{-1}D_\nu,
\qquad
\det K_\nu^{\rm eff}=\frac{\operatorname{Pf}(\mathbb A_N)}{\operatorname{Pf}(\mathbb A_H)},
\qquad
m_1m_2m_3=|\det K_\nu^{\rm eff}|.
\label{eq:bl-neutral-schur-complement}
\end{equation}
Writing $G_N:=(M_RR_N)^{-1}$, the affine-source Schur term is $\Phi_{XT}^{\rm Schur}=(A_L)_{XT}-E_X^TG_NE_T-E_T^TG_NE_X$ for $E_Y:=(D_\nu)_Y-(M_RR_N)_YG_ND_\nu$; it introduces no independent mixed tensor. On the restricted placement $M_RR_N=H_0+\varepsilon H_1$, $D_\nu=D_0+\varepsilon D_1$, put $E_K=D_1-H_1H_0^{-1}D_0$, $Q(\varepsilon)=E_K^T(H_0+\varepsilon H_1)^{-1}E_K$, and $T_n=E_K^T(H_0^{-1}H_1)^nH_0^{-1}E_K$. If $Q_0=T_0$ is invertible,
\begin{equation} Q(\varepsilon)^{-1}=Q_0^{-1}+\varepsilon Q_0^{-1}T_1Q_0^{-1},\qquad T_{n+1}=T_1Q_0^{-1}T_n. \label{eq:bl-neutral-affine-inverse-closure} \end{equation}
Away from the heavy determinant walls, and after subtraction of any separately typed direct nonlinear $A_L$ contribution, the second and third nonlinear jets fix the full Schur tail, with blind fourth-jet holdout $\partial_\varepsilon^4K_\nu^{\rm eff}(0)=-24T_1Q_0^{-1}T_1$. If $Q_0$ is singular, $d_\nu:=\rank[T_{i+j}]_{i,j=0}^2\leq3$ is the minimal observable heavy degree; standard realization theory \cite{Antoulas2005LargeScale} gives the observable quotient up to similarity. Thus $d_\nu=3$ recovers the full generalized heavy generator, while $d_\nu<3$ does not. The relay coefficients and actual Parent mixed source remain separately typed until the corresponding same-object source jets are exported. Conditionally on the literal same-carrier point/line-Real incidence, the zeroth-order frame is fixed. An actual neutral Parent matrix must select TM1 upstream. On the nondegenerate strict branch the marked-TBM nearest-$Z_2$ distances scale as $d_{\rm TM1}=O(\lambda_{\rm cen}^2)$ versus $O(\lambda_{\rm cen})$ for the alternatives; with Frobenius error $\varepsilon_M$, the observed gap to either alternative must exceed $2\varepsilon_M$. The retained neutral relay is
\begin{equation}
U_{\rm TBM}=\begin{pmatrix}\sqrt{2/3}&1/\sqrt3&0\\-1/\sqrt6&1/\sqrt3&1/\sqrt2\\-1/\sqrt6&1/\sqrt3&-1/\sqrt2\end{pmatrix},\qquad
U_\nu=U_{\rm TBM}R_{12}(\beta_\nu)R_{13}(\gamma_\nu)R_{23}(\alpha_\nu,\phi_\nu),
\label{eq:bl-pmns-rotation-convention}
\end{equation}
with $\alpha_\nu=(2/\sqrt3)\widehat\Xi_{23}\lambda_{\rm cen}$, $\gamma_\nu=-(13/7)\lambda_{\rm cen}^2$, and $\beta_\nu=-(7/5)\lambda_{\rm cen}^3$. Exact terminal attachment gives $\widehat\Xi_{23}^2=3/8$, hence $|\alpha_\nu|=\lambda_{\rm cen}/\sqrt2$ on the exact branch; the phase remains downstream of the same marked orientation/connection gate. Residual symmetry types the branch and hierarchy but does not supply $-13/7$ and $-7/5$, which remain microscopic Parent-response targets. For actual Takagi output, the first two rank-one projectors are scored by minimum maximal operator-norm distance to the strict two-projector family over the preregistered $\lambda_{\rm cen}$ interval and nuisance phase; the score is stable by the maximal projector error. It contains the marked-Real and rational constraints, while the terminal $1/\sqrt2$ condition is dependent and angle extraction a checksum. A finite relay jet is retained only with its flow-error certificate: if microscopic and relay paths differ by integrated operator norm $\mathfrak R_{\rm flow}$, each modulus changes by at most $\mathfrak R_{\rm flow}$, each squared modulus by at most $2\mathfrak R_{\rm flow}$, and a Jarlskog invariant by at most $4\mathfrak R_{\rm flow}$. Thus finite-jet closure is never promoted to an unrestricted analytic endpoint.\\
~\\
On the strict relay, put $x=\sin^2\theta_{13}$ and $y=\sin^2\theta_{23}$. To keep the longitudinal anchor independent of atmospheric data, the Daya Bay result $\sin^22\theta_{13}=0.0851\pm0.0024$ \cite{DayaBay2023Theta13} is used. On the physical branch this gives $x_R=0.0217479744$. Equation~\eqref{eq:bl-pmns-rotation-convention} then gives
\begin{equation}
\lambda_R=0.222439591087,\qquad y_R=0.469639881973,\qquad \left.\frac{dy}{dx}\right|_R=-0.03870477110.
\label{eq:bl-reactor-atmospheric-holdout}
\end{equation}
Exact propagation of the quoted $1\sigma$ reactor interval gives $0.469617440\leq y\leq0.469666105$ before microscopic flow error. From the finite-relay certificate, $|\delta y|\leq2(1+y_R)\mathfrak R_{\rm flow}/(1-x_R-2\mathfrak R_{\rm flow})$, with leading coefficient $3.00462$. A numerical audit on $0.20\leq\lambda_{\rm cen}\leq0.27$ finds one stationary point, $\lambda_{\rm foc}=0.241563238$ with $y_{\rm foc}=0.469541523$; the shallow minimum accounts for the weak longitudinal sensitivity. The reactor-propagated interval is reported separately from the microscopic flow budget.

\subsection{Saturation manifold, calibration and falsifiers}
\label{subsec:minimal-branch-numerical-benchmarks}

The trace-completed family readout and the global inverse charts used below are
\begin{equation}
A(u):=A_0-\frac32u,\qquad B_{\rm fl}(u):=B_0+\frac{40}{49}u,\qquad \lambda_{\rm cen}(u):=\frac{B_{\rm fl}(u)}{A(u)+B_{\rm fl}(u)}.
\label{eq:bl-trace-completed-flavor-readout}
\end{equation}
\begin{equation} u=\frac{2(7-9r_{WZ})}{49(1-r_{WZ})}=\frac{6(9\lambda_{\rm cen}-2)}{67\lambda_{\rm cen}+80}=\frac{18(15\eta_h-7)}{245(3\eta_h+2)}. \label{eq:bl-cell-inverse-coordinates} \end{equation}
Since $\mathcal M_{\rm alg}=[0,2/7)$, the $W/Z$ chart requires $0<r_{WZ}\leq7/9$, equivalently $m_W\leq m_Z\sqrt{7/9}$; this is a parameter-free falsifier of the selected Standard-Model solution branch.
The traceful scalar checksum is
\begin{equation}
\mathcal R_{\rm tr}:=\frac{db_{\rm tr}}{da_{\rm cen}}+\frac{80}{147}=0.
\label{eq:alena-traceful-scalar-incidence}
\end{equation}

\begin{proposition}[Reduced strict saturation packet]
\label{prop:bl-reduced-saturation-packet}
On a Parent image satisfying the realization gates of Section~\ref{sec:alena-codazzi-realization}, including the global Higgs and family-orbit incidences, let $\mathscr L_{\rm red}$ contain reduced PLSR, the active family valuation \eqref{eq:bl-hr-rees-valuation-lock}, the contraction relay of Corollary~\ref{cor:bl-canonical-jacobi-second-response}, and the norm-seven packet. Then $\rho_{\rm RP}=1$ is dependent. The fold fixes $u_*$ by \eqref{eq:alena-parent-selector-pushforward} with transport \eqref{eq:alena-parent-to-u-quadrature}; the trace lift is fixed by $B_{\rm fl}(0)=B_0$ and \eqref{eq:alena-family-scalar-characteristic-lock} through \eqref{eq:bl-trace-completed-flavor-readout}, retaining \eqref{eq:alena-traceful-scalar-incidence} as a checksum. Theorem~\ref{thm:bl-elementary-terminal-closure} and the same contraction provenance fix terminal/source-Hodge data, and Proposition~\ref{prop:alena-axial-associated-bundle} the neutral line. Physical orientation retains its marked-cycle incidence. With the two arithmetic norm-seven orientations as the discrete $\sigma_7^{\rm ar}$ fiber,
\begin{equation}
\mathcal M_{\rm sat}^{\rm red}\simeq\mathcal M_{\rm alg}\times\{+1,-1\}_{\sigma_7^{\rm ar}},\qquad \mathcal M_{\rm alg}=[0,2/7).
\label{eq:bl-reduced-saturation-manifold}
\end{equation}
The labeled normal-ordering neutrino formulas are restricted to $\mathcal M_{\rm NO}=[0,10/49)$ and remain conditional on the declared neutral mass-shape completion; they are excluded from the core no-retune count until its upstream factor is derived. Each chart in \eqref{eq:bl-cell-inverse-coordinates} is globally injective on $\mathcal M_{\rm alg}$, so one predeclared global chart observable fixes $u$ without sector-dependent retuning.
\end{proposition}
In the declared twelve-coordinate scale-free packet the regular local codimension is eleven after definitional identities. The Alena action is used here as a general realization class: one predeclared chart datum selects the longitudinal member identified with the Standard-Model realization, and that datum is not counted again as a prediction. All independently sourced transverse relations remain no-retune tests. A certified $u_{\rm Parent}$ from \eqref{eq:alena-parent-static-u}, or from another independently typed centered-response realization of \eqref{eq:alena-parent-to-u-quadrature}, replaces the calibration and adds one longitudinal test. For a one-dimensional chart with tangent $t=dY/d\lambda$ and calibration variance $\sigma_\lambda^2$, the first-order propagated covariance is $C_{\rm cal}^{(1)}=\sigma_\lambda^2tt^{\mathsf T}$ and has rank one. For experimental covariance $C$, the transverse Gaussian score uses
\begin{equation}
J_R^{\mathsf T}(J_RCJ_R^{\mathsf T})^{-1}J_R
=
C^{-1}-\frac{C^{-1}tt^{\mathsf T}C^{-1}}{t^{\mathsf T}C^{-1}t}.
\label{eq:bl-global-normal-projector}
\end{equation}
Calibration, branch selection, transverse holdouts and dependent diagnostics are separated in Subsection~\ref{app:experimental-scorecard}.

The exact Parent scale orbit leaves one absolute scale undetermined. An independently supplied protected primitive-action normalization may close this scale through $S_{\rm prim}=4\pi^2\kappa_\Omega$; no independent constructor for $S_{\rm prim}$ is assumed here. On the restricted normal-ordering branch put $x_\nu:=\sqrt{m_2/m_1}$, $y_\nu:=\sqrt{m_3/m_1}$ and $w_i:=|(U_\nu)_{ei}|^2$. After one atmospheric splitting $\Delta$ is supplied, the neutral scale-outs are
\begin{equation}
\frac{\sum_i m_i}{\sqrt\Delta}=\frac{1+x_\nu^2+y_\nu^2}{\sqrt{y_\nu^4-1}},
\qquad
\frac{m_\beta}{\sqrt\Delta}=\sqrt{\frac{w_1+w_2x_\nu^4+w_3y_\nu^4}{y_\nu^4-1}}.
\label{eq:bl-neutral-scale-calibrated-holdouts}
\end{equation}
For $A_1:=w_1/\sqrt{y_\nu^4-1}$, $A_2:=w_2x_\nu^2/\sqrt{y_\nu^4-1}$ and $A_3:=w_3y_\nu^2/\sqrt{y_\nu^4-1}$, the unresolved Majorana phases obey
\begin{equation}
\frac{m_{\beta\beta}^{\max}}{\sqrt\Delta}=A_1+A_2+A_3,
\qquad
\frac{m_{\beta\beta}^{\min}}{\sqrt\Delta}=\max\!\left\{2\max_iA_i-\sum_iA_i,0\right\}.
\label{eq:bl-majorana-phase-envelope}
\end{equation}
Absolute normalization, proton/heavy matching and threshold completion remain downstream; the comparison literature is \cite{Senjanovic2010ProtonDecayGUT}, \cite{HenningLuMurayama2018CDE}, \cite{DittmaierSchuhmacherStahlhofen2021}, and \cite{Djouadi2008Higgs}. The quarter-action neutral completion remains an ansatz and is not used in the no-retune count.

\subsection{Experimental holdouts and evidence accounting}
\label{app:experimental-scorecard}

\paragraph{Protocol and source ledger.}
The sole continuous CAL is the RPP-2026 $W/Z$ ratio; an independently certified $u_{\rm Parent}$ replaces it. CP branch selection consumes only the binary reporting sign $\sigma_7^{\rm rep}=\pm1$. Tier~1 contains direct/limited-assumption holdouts, Tier~2 Standard-Model/unitarity-conditioned coordinates, and Tier~3 completion diagnostics; calibration and branch inputs are not rescored. The reactor-to-atmospheric row is an alternative one-anchor forecast. Published multidimensional likelihoods are used when available; otherwise marginal pulls and diagonal scores remain diagnostics. Electroweak, CKM, Higgs, and charged-lepton reference values follow RPP 2026 \cite{PDG2026RPP}; normal-ordering global coordinates follow NuFIT 6.1 \cite{NuFIT2025v61}. Independent oscillation checks are represented by JUNO \cite{JUNO2026FirstOsc}, Daya Bay \cite{DayaBay2023Theta13}, and T2K-NOvA \cite{T2KNOvA2025Joint}; $B$-physics by HFLAV \cite{HFLAVSummer2025}. Prospective charged-lepton and oscillation sensitivities are represented by \cite{Fu2024Ditauonium}, \cite{HyperKamiokande2026Sensitivity}, and \cite{DUNE2022Sensitivity}.\\
~\\
The $W/Z$ calibration gives $\lambda_{WZ}=0.225452432878$ and $u_{WZ}=0.00183409\pm0.00024988$, placing the solution $7.34\sigma$ inside this branch boundary. For $Y=(Y_q,Y_\nu)$ with $Y_q=(s_{12}^q,s_{23}^q,s_{13}^q)$ and $Y_\nu=(s_{12,\nu}^2,s_{13,\nu}^2,s_{23,\nu}^2)$, the tangent components are $Y'_q=(1,0.3805923,0.0494512)$ and $Y'_\nu=(-0.2975254,0.2904966,-0.0089577)$, so the first-order calibration covariance is rank one. Direct quark magnitudes, Daya Bay and peer-reviewed JUNO solar data reconstruct a consistent longitudinal coordinate: their direct-flavor combination gives $\lambda_{\rm flav}^{\rm dir}=0.224156\pm0.000749$, $1.49\sigma$ from the $W/Z$ value, with internal $p=0.698$. Including $W/Z$ only as a consistency diagnostic gives $\chi^2_{\rm long}=2.942$ for three degrees of freedom, $p=0.401$. Global-fit CKM and NuFIT coordinates remain Tier~2 diagnostics and are not used to redefine the calibration.

\begingroup
\scriptsize
\setlength{\tabcolsep}{2.5pt}
\renewcommand{\arraystretch}{1.08}
\begin{longtable}{@{}p{0.20\textwidth}p{0.27\textwidth}p{0.27\textwidth}p{0.18\textwidth}@{}}
\caption{Calibration and principal tests on the reduced one-dimensional solution manifold.}\label{tab:experimental-scale-free-scorecard}\\
\toprule
Observable & Model / inferred value & Experimental datum & Role / pull\\
\midrule
\endfirsthead
\toprule
Observable & Model / inferred value & Experimental datum & Role / pull\\
\midrule
\endhead
EW: $r_{WZ}$ & $0.776662701\pm0.000152683$ & $0.776662701\pm0.000152683$ & CAL; $7.34\sigma$ inside branch boundary\\
CKM first-row geometry & $\mathcal U_1=1$; $\lambda_{\rm row1}=0.2254442\pm0.0007065$ & $\mathcal U_1^{\rm exp}=0.9984\pm0.0007$ & Tier~1 HOLDOUT; radial deficit $2.29\sigma$, longitudinal shift $-0.01\sigma$\\
CKM CP sign & $\sigma_7^{\rm rep}=+1$ & $\operatorname{sgn}J_q=+1$ in the marked CKM convention & REPORTING-BRANCH; one bit, not a Wilson PASS\\
UT phase packet & $\sin2\beta=0.710159$, $\gamma=65.325^\circ$, $\alpha=92.051^\circ$ & $0.710\pm0.011$, $66.4^{+2.7}_{-2.8}{}^\circ$, $84.1^{+3.7}_{-3.0}{}^\circ$ & Tier~1; $0.014\sigma$, $-0.38\sigma$, $+2.15\sigma$\\
JUNO solar & $\sin^2\theta_{12}=0.308722\pm0.000131$ & $0.3092\pm0.0087$ & Tier~1 HOLDOUT; $-0.055\sigma$\\
Reactor $\to$ atmospheric & $y_R=0.4696399$, $1\sigma$ reactor band $[0.4696174,0.4696661]$ & Daya Bay: $\sin^22\theta_{13}=0.0851\pm0.0024$; $y_{\rm exp}=0.470\pm0.015$ & HOLDOUT; alternative longitudinal anchor; separate flow budget above\\
Charged leptons: $m_\tau$ [MeV] & $1776.969027$ & $1776.93\pm0.09$ & Tier~1 HOLDOUT; $+0.43\sigma$\\
Higgs primitive: $m_h^{\rm prim}$ [GeV] & $125.343\pm0.103$ & $125.13\pm0.11$ & Tier~3 COMPLETION-DIAG; $+1.41\sigma$\\
\bottomrule
\end{longtable}
\endgroup

\paragraph{Transverse structural tests.}
The first-row packet is counted once: its radial norm gives the structural test in Table~\ref{tab:experimental-scale-free-scorecard}, while the normalized direction is consistent with the $W/Z$ solution with $p=0.349$; profiling the same packet on the one-dimensional manifold gives transverse $p=0.0565$. The direct unitarity-triangle pulls in Table~\ref{tab:experimental-scale-free-scorecard} are not combined without their joint covariance.

\paragraph{Longitudinal consistency and dependent readings.}
As a longitudinal cross-check using overlapping CKM data, the direct $(|V_{us}|,|V_{cb}|,|V_{ub}|)$ packet gives $\lambda_q^{\rm dir}=0.224391\pm0.000799$, $1.16\sigma$ below the $W/Z$ coordinate, with transverse $p=0.356$; it is not counted as an additional independent holdout. Dependent projections carry no additional evidential weight. The CKM inverse gives $u_{\rm CKM}=0.0016741\pm0.0003854$, $0.35\sigma$ from $u_{WZ}$, and $|J_q|=(3.16291\pm0.037)\times10^{-5}$ against $(3.16\pm0.12)\times10^{-5}$. Global-fit CKM/NuFIT coordinates remain Tier~2 diagnostics. The neutral readings are $\delta_{\rm PMNS}\in\{206.877^\circ,153.123^\circ\}\pm0.044^\circ$, $Q_\ell=2/3$ and $|J_{\rm PMNS}|=0.01531842\pm0.00001660$. The W/Z-calibrated Daya Bay projection is at $+1.35\sigma$; under reactor anchoring it is not rescored.

Conditional on the neutral mass-shape bridge, the restricted normal-ordering branch gives $(m_1,m_2,m_3)/\sum m_\nu=(0.0119389,0.1443804,0.8436806)$, $m_\beta/\sum m_\nu=0.1499360\pm0.0001080$, and $m_{\beta\beta}/\sum m_\nu\in[0.0164225,0.0707086]$. Using one NuFIT atmospheric splitting only as the dimensionful scale in \eqref{eq:bl-neutral-scale-calibrated-holdouts} gives $\sum_i m_i=59.518\,{\rm meV}$ and $m_\beta=8.924\,{\rm meV}$; \eqref{eq:bl-majorana-phase-envelope} gives $0.977\leq m_{\beta\beta}\leq4.208\,{\rm meV}$. The same completion gives $\Delta m_{21}^2/\Delta m_{31}^2=0.0290916\pm0.00056$; the 207.2-day JUNO monitor \cite{JUNO2026ImprovedOsc} gives $0.0294460\pm0.0004356$, a $-0.50\sigma$ difference. Because its physical $\Delta m_{31}^2$ minimum uses an external constraint, this remains a Tier~3 completion diagnostic. These are SCALE-OUT/FUTURE readings rather than zero-input mass predictions.

The conditional restricted-NO scale-out above gives $\sum m_\nu=59.518\,{\rm meV}$. The DESI DR2 $\Lambda$CDM bound \cite{DESI2025DR2Neutrino}, $64.2\,{\rm meV}$ at $95\%$ CL, leaves $4.682\,{\rm meV}$ headroom, while the quoted $w_0w_a$CDM bound relaxes to $163\,{\rm meV}$; this is retained as a cosmology-model-dependent Tier~3 monitor. The quarter-action determinant-suppressed completion remains an ansatz, with benchmark sum $59.900342\,{\rm meV}$ compared also with \cite{DESI2025Cosmology}; the unit-determinant pure type-I benchmark $\sum m_\nu\simeq0.1474\,{\rm eV}$ is excluded on the quoted $\Lambda$CDM bound and remains allowed by the quoted evolving-dark-energy bound. Common-scale conversion follows the RPP convention \cite{PDG2026RPP}. On the published Standard-Model running grid \cite{AntuschHinzeSaad2026Running}, no displayed benchmark realizes all three structural shells simultaneously; this remains a COMPLETION-DIAG statement until joint covariance, high-scale kinetic law, and threshold data are supplied.

\section{Discussion and Conclusions}
\label{sec:conclusions}
\label{sec:discussion}

The central result is that protected boundary data of the timelike Codazzi defect determine, within the declared reconstruction class, a finite $3+2$ carrier before a microscopic action is selected. The primitive exact complex, determinant/exterior gauge and chiral packet, and norm-seven family packet are descendants of one reconstructed source. The exterior packet has the standard algebraic AIII-$SU(5)$-Spin$(10)$ envelope, while the norm-seven $3\oplus\bar3$ splitting and cameral orientation are the eigenspaces and sign of one canonical real augmentation complex structure. These protected statements are independent of the later longitudinal calibration.\
~\\
A constrained same-domain Alena-Codazzi realization is constructed explicitly. Its selector terminates at one intrinsic event whose charge minimum and first simple invariant-sector fixed-$Q$ loss are dependent readings; for nontrivial soft-coupled readouts the same terminal carries a rank-one critical-response pole. The local $A_2$ packet has the standard Frobenius/Landau-Ginzburg and finite twisted-period descriptions. Verification of actual root transport is reduced to a determinant lock, three Fricke trace coordinates and the frozen multiplication incidence; the norm-seven orientation remains a scalar same-loop Wilson gate. Existence is established within the stated regularized-domain class; uniqueness and microscopic naturalness of the selector-core realization are not claimed.\
~\\
The longitudinal microscopic attachment is formulated through the native $C$ port. On the exact-five moment stratum, four moments reconstruct the minimal five-dimensional operator up to hidden unitary, and primitive multiplication identifies its hidden two-space with the protected $W$. The selector core does not fix the normalized centered response. A same-object literal static susceptibility with common gain would give the blind coordinate \eqref{eq:alena-parent-static-u}; otherwise one predeclared chart datum selects the Standard-Model member of the one-dimensional saturation manifold of the same general action. The direct-data scorecard is separated from Standard-Model-conditioned coordinates and completion diagnostics.\
~\\
The neutral affine Schur placement has a sharper minimal realization: on the invertible branch its inverse tail is affine and higher nonlinear jets are dependent, while the singular branch is classified by the block-Hankel observable degree. The $q=5$ sector remains a literal five-block Parent anchor, but one anchor plus derivative bounds certifies its inertia neighbourhood; after all sector and essential gaps pass, reduced coercivity gives local Hamiltonian spectral stability without a separate finite Krein-collision analysis. Leading moving-dark leakage is already reduced to one scalar Herglotz channel; its absolute normalization and threshold class remain open. The remaining problems concern literal operator export, transport and continuum completion: actual endpoint/Green data, neutral source jets or nodes, root/Wilson holonomies, the $q=5$ and essential gaps, nonsinglet thresholds, the off-axis family source and absolute normalization remain to be supplied. Nonlinear orbital stability remains a separate theorem burden. These steps do not require a further reconstruction of the finite carrier and do not alter the existence result proved here.

\section*{Statements}
~\\
Author has no relevant financial or non-financial interests to disclose.\\
~\\
Author did not receive support from any organization for the submitted work.\\
~\\
No new datasets were generated in this study. Numerical comparisons use published values and results from the sources cited in Subsection~\ref{app:experimental-scorecard}.\\
~\\
During the preparation of this manuscript, the author used generative AI tools for language editing, formatting, consistency checks, and organization of selected passages. These tools were not used to generate research data. All mathematical statements, citations, and scientific claims were reviewed, verified, and accepted by the author, who takes full responsibility for the final manuscript.\\

\bibliography{ATCodazzi}

\end{document}